\documentclass[journal]{IEEEtran}

\usepackage{amsmath,amssymb,amsfonts}
 
\usepackage{amsthm}
\usepackage{algorithm}
\usepackage{algorithmic}
\usepackage{graphicx}
\usepackage{mathtools}
\usepackage{hyperref}
\hypersetup{hidelinks}
\usepackage{textcomp}
\usepackage[style=ieee, sorting=nyt]{biblatex}
\usepackage{tikz}
\usetikzlibrary{positioning,arrows.meta,calc}
\graphicspath{{./figs/}}
\newtheorem{theorem}{Theorem}

\newtheorem{assumption}{Assumption}
\newtheorem{lemma}{Lemma}

\newtheorem{proposition}{Proposition}

\newtheorem{definition}{Definition}
\newtheorem{remark}{Remark}

\def\ddefloop#1{\ifx\ddefloop#1\else\ddef{#1}\expandafter\ddefloop\fi}
\def\ddef#1{\expandafter\def\csname fr#1\endcsname{\ensuremath{\mathfrak{#1} } } }
\ddefloop ABCDEFGHIJKLMNOPQRSTUVWXYZ\ddefloop
\def\ddef#1{\expandafter\def\csname bb#1\endcsname{\ensuremath{\mathbb{#1} } } }
\ddefloop ABCDEFGHIJKLMNOPQRSTUVWXYZ\ddefloop
\def\ddef#1{\expandafter\def\csname bf#1\endcsname{\ensuremath{\mathbf{#1} } } }
\ddefloop ABCDEFGHIJKLMNOPQRSTUVWXYZabcdefghijklmnopqrstuvwxyz\ddefloop
\def\ddef#1{\expandafter\def\csname bs#1\endcsname{\ensuremath{\boldsymbol{#1} } } }
\ddefloop ABCDEFGHIJKLMNOPQRSTUVWXYZabcdefghijklmnopqrstuvwxyz\ddefloop
\def\ddef#1{\expandafter\def\csname sf#1\endcsname{\ensuremath{\mathsf{#1} } } }
\ddefloop ABCDEFGHIJKLMNOPQRSTUVWXYZ\ddefloop
\def\ddef#1{\expandafter\def\csname c#1\endcsname{\ensuremath{\mathcal{#1} } } }
\ddefloop ABCDEFGHIJKLMNOPQRSTUVWXYZ\ddefloop
\def\ddef#1{\expandafter\def\csname h#1\endcsname{\ensuremath{\hat{#1} } } }
\ddefloop ABCDEFGHIJKLMNOPQRSTUVWXYZ\ddefloop
\def\ddef#1{\expandafter\def\csname dt#1\endcsname{\ensuremath{\dot{#1} } } }
\ddefloop ABCDEFGHIJKLMNOPQRSTUVWXYZ\ddefloop
\def\ddef#1{\expandafter\def\csname t#1\endcsname{\ensuremath{\tilde{#1} } } }
\ddefloop ABCDEFGHIJKLMNOPQRSTUVWXYZ\ddefloop

\newcommand{\T}{\intercal}

\DeclareMathOperator*{\minimize}{minimize}

\DeclareMathOperator*{\argmax}{argmax}

\DeclareMathOperator{\laspan}{\mathrm{span}}
\DeclareMathOperator{\larank}{rank}

\newcommand{\lie}{\mathrm{Lie}}
\newcommand{\rank}{\mathrm{rank}}

\def\BibTeX{{\rm B\kern-.05em{\sc i\kern-.025em b}\kern-.08em
T\kern-.1667em\lower.7ex\hbox{E}\kern-.125emX}}

\begin{document}

\title{A Controllability Perspective on Steering Follow-the-Regularized-Leader Learners in Games}
\author{Heling Zhang, Siqi Du, and Roy Dong
    \thanks{This work was supported by the National Science Foundation under Grant CCF 2236484.}
    \thanks{H. Zhang is with the Department of Electrical and Computer Engineering at Illinois Grainger Engineering, University of Illinois Urbana-Champaign. (email: hzhng120@illinois.edu).}
    \thanks{S. Du and R. Dong are with the Department of Industrial and Enterprise Systems Engineering at Illinois Grainger Engineering, University of Illinois Urbana-Champaign. (emails: \{siqidu3,roydong\}@illinois.edu).}
}

\maketitle

\begin{abstract}
Follow-the-regularized-leader (FTRL) algorithms have become popular in the context of games, providing easy-to-implement methods for each agent, as well as theoretical guarantees that the strategies of all agents will converge to some equilibrium concept (provided that all agents follow the appropriate dynamics). However, with these methods, each agent ignores the coupling in the game, and treats their payoff vectors as exogenously given. In this paper, we take the perspective of one agent (the controller) deciding their mixed strategies in a finite game, while one or more other agents update their mixed strategies according to continuous-time FTRL. Viewing the learners' dynamics as a nonlinear control system evolving on the relative interior of a simplex or product of simplices, we ask when the controller can steer the learners to a target state, using only its own mixed strategy and without modifying the game's payoff structure. 

For the two-player case we provide a necessary and sufficient criterion for controllability based on the existence of a fully mixed neutralizing controller strategy and a rank condition on the projected payoff map. For multi-learner interactions we give two sufficient controllability conditions, one based on uniform neutralization and one based on a periodic-drift hypothesis together with a Lie-algebra rank condition. We illustrate these results on canonical examples such as Rock-Paper-Scissors and a construction related to Brockett's integrator.
\end{abstract}

\section{Introduction}
Follow-the-regularized-leader (FTRL) algorithms have become standard techniques for handling complicated strategic interactions in multi-agent environments~\cite{cesa09, fudenberg22}. From a practical standpoint, there is a compelling case for deploying FTRL in real-world systems: it provides an easy-to-implement, computationally efficient update rule for each agent that relies strictly on locally observed payoff feedback~\cite{arora12, toonsi23}. Furthermore, it is supported by a rich theoretical literature ensuring that, provided all agents follow appropriate dynamics, the population's strategies will converge to established equilibrium concepts~\cite{foster97, hart00, cai25}. As a result, FTRL is frequently utilized to design autonomous agents navigating complex, repeated interactions.

However, this decentralized simplicity comes with a fundamental structural assumption. By employing FTRL, each agent inherently ignores the strategic coupling of the underlying game. Rather than recognizing the interaction as a closed-loop feedback system where their own actions influence the future behavior of others, the learner treats their incoming sequence of payoff vectors as exogenously given~\cite{papadimitriou19, toonsi25}. To an FTRL agent, the multi-agent environment is perceived merely as a fluctuating, open-loop landscape to be optimized against~\cite{cheung21}. If an entire population of interacting agents blindly adopts this uncoupled learning paradigm, it can introduce new systemic vulnerabilities and opportunities for strategic exploitation~\cite{braverman17, deng19}.

Suppose a single, sophisticated agent is aware of this behavioral structure. Recognizing that their opponents are predictably driven by FTRL updates, this model-aware agent no longer needs to myopically optimize their immediate payoff~\cite{lauffer22, mansour22}. Instead, they can actively shape the exogenous payoffs observed by the learners, treating the evolving mixed strategies of the population as a dynamical system to be manipulated~\cite{brown24}. The natural question then arises: where could they steer the system? 

The practical implications of such reachability are significant. For instance, in automated financial markets or algorithmic pricing, a strategic participant could manipulate FTRL-driven competitors to drive the market into profitable, out-of-equilibrium pricing configurations~\cite{braverman17}. Similarly, in intelligent infrastructure, an adversarial or central node could steer independent routing protocols to induce targeted congestion or enforce globally optimal traffic flows~\cite{zhang24}. 
These examples motivate a controllability viewpoint: when the
game and learning dynamics satisfy suitable conditions, a
model-aware strategic player may be able to steer the induced
learning dynamics toward selected states.

This paper studies this steering problem as a controllability problem. We consider a finite game with a distinguished controller and one or more learners that follow continuous-time FTRL. Interpreting the controller’s mixed strategy as the control input and the learners’ mixed strategies as the state yields a nonlinear control system evolving on the relative interior of the simplex, or a product of simplices if there are multiple learners. The restriction to the relative interior is based on the following consideration: under common learning flows (e.g., replicator dynamics), the boundary is invariant, so interior initial conditions cannot be driven to the boundary in finite time~\cite{weibull95}. 

Closest to our work are recent papers on steering learners in games, but the
control channel, objective, and analysis are different. In contrast to work that steers no-regret learners through external payments or dynamic incentives, we keep the game and payoff structure fixed and allow the controller to act only through its own mixed strategy. In contrast to repeated-game formulations that seek to drive a learner toward a Stackelberg outcome while learning unknown payoffs, we study a model-based controllability question for continuous-time FTRL dynamics on the relative interior of the simplex. Our contribution is a control-theoretic formulation of this steering problem, together with an exact two-player controllability criterion and two sufficient controllability conditions for multi-learner interactions based on neutralization, periodic drift, and Lie-algebraic rank tests.

The contributions of this paper are threefold. First, we formulate steering
continuous-time FTRL learners as a nonlinear controllability problem on the
relative interior of a simplex or product of simplices, with the controller's mixed strategy serving as the admissible control input. Second, in the two-player case we derive an exact controllability criterion based on a fully mixed
neutralizing strategy and a projected-payoff rank condition. Third, in the
multi-learner case we derive two sufficient controllability conditions, one via
uniform neutralization and one via periodic drift together with a Lie-algebra
rank condition. We illustrate our results on canonical examples such as Rock-Paper-Scissors and a construction related to Brockett’s integrator.

\section{Related Works}
\subsection{No-Regret Learning in Multi-Agent Settings}

In multi-agent settings, no-regret learning rules like FTRL do not directly consider the coupling between agents~\cite{toonsi23}. Rather than explicitly modeling their opponents, learners update their strategies based solely on their own historical payoff observations. They effectively treat the game's endogenous strategic interactions as exogenous environmental signals.

Even so, no-regret learning remains attractive because it is decentralized, requires only limited feedback, and gives each agent a clear long-run performance guarantee against the realized behavior of others. The main technical nuance is that these guarantees are typically about average behavior rather than the last iterate: time-averaged play can exhibit meaningful convergence properties even when the period-by-period strategies continue to move. This makes no-regret dynamics useful when long-run average performance is the relevant objective, but less satisfactory when one needs stable pointwise convergence of strategies.

To address this non-convergence, an active branch of algorithmic research modifies the learning rules to achieve last-iterate convergence, utilizing methods such as optimistic gradient extrapolation~\cite{tsuchiya24}, continuous action perturbations~\cite{abe22}, or black-box reductions~\cite{cai25}. However, these approaches require the ability to dictate the agents' internal algorithms. Our work considers an alternative setting: assuming the agents' standard FTRL dynamics are fixed, we investigate whether a single strategic player can steer the resulting system purely through their own actions.

\subsection{Strategizing Against Learners}

Because FTRL dynamics do not explicitly account for the closed-loop effects in a multi-agent setting, a large body of literature focuses on how to strategically exploit FTRL agents. This research usually frames the interaction as an asymmetric dynamic game, where a sophisticated optimizer adjusts its strategy to take advantage of a naive learner over time~\cite{mansour22, deng19}. The ultimate goal here is almost entirely payoff-oriented. Researchers frequently rely on Stackelberg formulations to compute commitment strategies that extract far more cumulative utility than a standard simultaneous-play Nash equilibrium would allow. Recent work also shows that this viewpoint faces strong computational barriers in general settings: unless P=NP, no polynomial-time optimizer can compute a near-optimal strategy against a learner running a standard no-regret algorithm such as multiplicative weights~\cite{assos25}.

While this work offers tight bounds on utility extraction and regret minimization, it treats the learner's evolving state simply as a vehicle for reaching a specific scalar payoff limit. This framework does not consider whether arbitrary coordinates on the non-Euclidean simplex can actually be reached, nor does it offer operational guidance for transient physical control. For instance, if a controller needs a population to temporarily adopt a specific, non-equilibrium mixed strategy just to avoid a catastrophic failure in a physical routing network, scalar payoff bounds are practically useless. By divorcing utility maximization from the mechanics of nonlinear state reachability, this literature largely ignores the underlying dynamics and state of the system.

\subsection{Steering via Mechanism Design}
Another active area of research investigates how central planners can guide learning populations toward stable or socially optimal states~\cite{canyakmaz24}. A common approach is to introduce dynamic incentives or side payments directly into the learners' utility functions to mitigate the oscillatory nature of FTRL dynamics~\cite{zhang24}.

Although incentive-based steering is effective, it requires a mediator with the administrative authority to actively modify the underlying payoff structure. In many decentralized or competitive settings, a standard player lacks the ability to unilaterally alter an opponent's utility function~\cite{manheim19}. Thus, while mechanism design demonstrates that a system can be guided by adjusting incentives, our work investigates whether learners can be steered to target states purely through strategic play within a rigidly fixed game topology.

\subsection{Game Dynamics and Geometric Control}
To analyze state reachability on the probability simplex, we utilize tools from geometric control theory. Specifically, we rely on the Chow-Rashevskii Theorem~\cite{chow39, rashevskii38} to establish local controllability, and we use~\cite[Theorem 1]{boscain21} to establish global controllability. For broader treatments of these standard tools, one can refer to~\cite{isidori85} or~\cite{agrachev04}.

The closest mathematical analogues to our approach appear in evolutionary biology and population dynamics. In those fields, continuous-time game equations like the replicator dynamics have been mapped onto the cotangent bundle of the probability simplex, allowing researchers to establish local controllability~\cite{raju20}. However, these prior studies typically model the control input as an exogenous environmental parameter or mutation rate~\cite{halder23}, rather than the actions of a single participant.

Our work adapts this methodology to matrix games by defining the control input strictly as the mixed strategy of a participating player. By introducing the concept of a neutralizing strategy to counteract the system's strong autonomous drift, we are able to map the algebraic tools of geometric control directly onto the primal manifold of the multi-agent interaction.

In summary, the existing literature influences learning dynamics in three main ways: by optimizing against learners for payoff, by modifying their incentives, or by introducing exogenous control signals. Our setting is different. We ask whether a standard participating player, without changing the learners' utility functions or internal algorithms, can use only its own mixed strategy to reach arbitrary interior states through the continuous-time FTRL dynamics. This perspective leads naturally to a control theoretic formulation of the problem and to corresponding controllability criteria.


\section{Notation}
 For any $N \in \bbN$, we use $[N]$ to represent $\{1, \dots, N\}$. Let $x_{1}, \dots, x_{n}$ be a sequence of $n$ objects; we denote by $x_{-i}$ the sequence with $x_{i}$ removed, i.e. $x_{-i} = (x_{1}, \dots, x_{i-1}, x_{i+1}, \dots, x_{n})$. Let $\Delta_{n} := \{ (x_{1}, \dots, x_{n}) | x_{i} \ge 0, \sum_{i=1}^{n} x_{i} = 1\}$ be the $n$-dimensional probability simplex. For a set $S$, we denote its \textbf{relative} interior by $S^{\circ}$, and its closure by $\overline{S}$. Let $\mathbf{1}_{n}$ be the $n$-dimensional all-ones vector, i.e. $\mathbf{1} = (1, 1, \dots, 1) \in \bbR^{n}$. We omit subscripts and write $\mathbf{1}$ when the dimension can be inferred from the context. We denote the concatenation of vectors $x \in \bbR^{n}$ and $y \in \bbR^{m}$ by $x\oplus y$. For a function $f: \cX \to \cY$ and a set $S \subseteq \cX$, we use $f|_{S}$ to represent the restriction of $f$ to $S$. 

For a vector field $g(x)$, we use $e^{tg}(x_0)$ to denote the time-$t$ flow of
$g$ starting from $x_0$. For a set of vector fields $\cG(x) = \{g_{1}(x), \dots, g_{n}(x)\}$, we denote the Lie algebra generated by $\cG(x)$ as $\lie(\cG)$. These notations will be further elaborated in the next section when we talk about geometric control.

For manifolds \(M,N\) and \(k \in \{0,1,\infty\}\), we write \(C^{k}(M,N)\) for the space of $k$ times continuously differentiable functions from $M$ to $N$. When \(N=\bbR\), we abbreviate this as \(C^{k}(M)\). For \(q \in M\), \(T_{q}M\) denotes the tangent space of \(M\) at \(q\). If \(X \in C^{1}(M,N)\), then \(DX(q):T_{q}M \to T_{X(q)}N\) denotes the differential of \(X\) at \(q\).

\vspace{2cm}

\section{Preliminaries}
\label{sec:preliminaries}
\subsection{Finite Games}
In this paper, we focus on finite games, i.e., normal-form games with finite action spaces. Formally, such games are defined as follows.
\begin{definition}[Finite Games]
    \label{def:finite-game}
    A finite game consists of the following:
    \begin{enumerate}
        \item a finite set of players indexed by $[N]$,
        \item for each player $i$, a finite set of actions $[n_{i}]$,
        \item for each player $i$, a utility function $r_i: [n_{1}] \times \dots \times [n_{N}] \to \bbR$ that quantifies the player's preference over the joint action profile of all players.
    \end{enumerate}
\end{definition}

The above definition naturally extends to mixed strategies. A mixed strategy of player $i$ is a vector $x_i \in \Delta_{n_i}$ representing a probability distribution over $[n_{i}]$, where $x_i(a)$ denotes the probability assigned to action $a \in [n_{i}]$. For convenience, let $\cA_i := [n_i]$ and let $\cA_{-i}$ denote the joint action set of all players other than $i$. Given a mixed strategy profile $(x_1, \dots, x_N) \in \Delta_{n_1} \times \dots \times \Delta_{n_N}$, the expected utility of player $i$ is
\begin{equation*}
    \begin{aligned}
        R(x_i, x_{-i}) 
        &= \sum_{a_i \in \cA_i} x_i(a_i) \sum_{a_{-i} \in \cA_{-i}} r_i(a_i, a_{-i}) \prod_{j \in [N] \setminus \{i\}} x_j(a_j) \\
        &= \langle x_i, p_i \rangle,
    \end{aligned}
\end{equation*}
where 
\begin{equation*}
    p_i = 
    \begin{pmatrix}
        \sum_{a_{-1} \in \cA_{-1}}
        r_i(1, a_{-i}) \prod_{j \in [N] \setminus \{i\}} x_j(a_j) \\
        \vdots \\
        \sum_{a_{-1} \in \cA_{-1}}
        r_i(n_i, a_{-i}) \prod_{j \in [N] \setminus \{i\}} x_j(a_j) \\
    \end{pmatrix}
    \label{eq:payoff-vector}
\end{equation*}
is the corresponding payoff vector of player $i$.
\subsection{Continuous-Time FTRL}
Learning in games studies how players adapt their strategies over time in response to payoff feedback generated by the evolving play of others. In this paper, we focus on learning rules in the following form,
\begin{equation*}
        x = f(y), \quad 
        \dot{y} = g(y, p),
\end{equation*}
where $x$ is the mixed strategy, $p$ is the payoff vector and $y$ is an auxiliary state that stores historical information about the environment.

One of the most studied families of such learning rules is the FTRL class of algorithms, which is given by
\begin{equation*}
    x = Q(y), \quad \dot{y} = p,
\end{equation*}
where
\begin{equation}
\label{eq:Q}
    Q(y) = \argmax_{x' \in \Delta_k} \langle x', y \rangle - h(x')
\end{equation}
for some regularizer $h$. In this paper, we make the following assumptions on $h$:
\begin{assumption}
\label{asm:regularizer}
    We assume
    \begin{enumerate}
        \item $h: \bbR^{n} \to \bbR$ is smooth in the sense that $h \in C^{\infty}$ on $\Delta_n^\circ$, and
        \item the Hessian $\nabla^{2}h(x)$ is positive definite for all $x \in \Delta_{n}^\circ$,
    \end{enumerate}
    where $\Delta_n^\circ$ is the relative interior of the probability simplex.
\end{assumption}

\subsection{Geometric Control}
In this subsection, we present our main theoretic tools from
geometric control.

Consider the nonlinear control system
\begin{equation} \label{eq:sys_x}
\Sigma_x: \quad \dot{x} = f(x, u), \quad x \in M, \quad u \in U,
\end{equation}
where $M$ is a connected smooth manifold (in our applications,
$M$ will be the relative interior of a simplex or a product of
simplices) and $U$ is the set of admissible controls. We assume
that for each fixed $u \in U$ the map $x \mapsto f(x, u)$ defines a
smooth vector field on $M$.

For a fixed control value $u \in U$ write $f_u(\cdot) := f(\cdot, u)$
for the corresponding vector field. We denote by $e^{tf_u}(x_0)$ the
time-$t$ flow of $f_u$ starting from $x_0$, i.e., the unique value at
time $t$ of the solution $x(\cdot)$ to $\dot{x} = f_u(x)$ with initial condition
$x(0) = x_0$ whenever the solution uniquely exists.

If the control signal $u(\cdot)$ is piecewise constant, taking values
$u_1, \dots, u_k \in U$ over successive time intervals of lengths
$t_1, \dots, t_k \ge 0$, then the resulting trajectory satisfies
\begin{equation*}
x(t_1 + \dots + t_k) = e^{t_k f_{u_k}} \circ \dots \circ e^{t_1 f_{u_1}}(x_0),
\end{equation*}
which we sometimes abbreviate
$x = e^{t_k f_{u_k}} \dots e^{t_1 f_{u_1}} x_0$.

Let
\begin{equation*}
\mathcal{F} := \{f_u(\cdot) : u \in U\}
\end{equation*}
denote the family of control vector fields. For $x_0 \in M$ and
$T \ge 0$, define the attainable set within time $T$ by
\begin{align*}
A_x(\le T, x_0) &= \bigg\{ e^{t_k f_{u_k}} \circ \dots \circ e^{t_1 f_{u_1}}(x_0) : \\
&\qquad k \in \mathbb{N}, u_i \in U, t_i \ge 0, \sum_{i=1}^k t_i \le T \bigg\},
\end{align*}
and the (forward) attainable set by $A_x(x_0) := \bigcup_{T > 0} A_x(\le T, x_0)$.

Based on these attainable sets, we define three forms of controllability for the system \eqref{eq:sys_x}:
\begin{enumerate}
    \item \textbf{Small-time local controllability (STLC):} The system is STLC if for every $x \in M$ and every $T > 0$, the state $x$ is in the interior of the attainable set within time $T$, i.e., $x \in \operatorname{int}(A_x(\le T, x))$.
    \item \textbf{Local controllability:} The system is locally controllable if for every $x \in M$, the state $x$ is in the interior of the forward attainable set, i.e., $x \in \operatorname{int}(A_x(x))$.
    \item \textbf{Global controllability:} The system is globally controllable (or simply \textit{controllable}) on $M$ if the attainable set from any state is the entire manifold, i.e., $A_x(x) = M$ for all $x \in M$.
\end{enumerate}
Throughout this paper, we will invoke controllability results for driftless
systems under a mild condition on the set of admissible controls. 

\begin{definition}[Proper Control Set]
We call a control set $U \subset \mathbb{R}^m$ \textbf{proper} if it is
convex, compact, and contains the origin in its relative interior.
\end{definition}

This condition guarantees that one can realize small control
variations in all directions within the affine subspace generated by $U$. In particular, if
$u_0 \in \Delta_m^\circ$ is fully mixed, then the translated set $V := \Delta_m - u_0$
is proper in the affine subspace $\{v \in \mathbb{R}^m : \langle v, \mathbf{1} \rangle = 0\}$ and
satisfies $0 \in V^\circ$.

For $C^1$ vector fields $X$ and $Y$ on $M$, their Lie bracket $[X, Y]$ is
the $C^0$ vector field defined (in local coordinates) by
\begin{equation*}
[X, Y](x) = DY(x)X(x) - DX(x)Y(x).
\end{equation*}
The Lie algebra generated by a family of smooth vector fields
$\mathcal{G}$, denoted $\operatorname{Lie}(\mathcal{G})$, is the smallest collection of vector fields
containing $\mathcal{G}$ and closed under finite linear combinations and
Lie brackets. Its evaluation at $q \in M$ is the subspace
\begin{equation*}
\operatorname{Lie}_q(\mathcal{G}) := \{X(q) : X \in \operatorname{Lie}(\mathcal{G})\} \subseteq T_q M.
\end{equation*}

\begin{definition}[Bracket Generating]
We say that a family of vector fields $\mathcal{G}$ is bracket generating on $M$ if $\operatorname{Lie}_q(\mathcal{G}) = T_q M$ for every $q \in M$.
\end{definition}


We will use (without proof) the following classical facts.

\begin{proposition}[Chow-Rashevskii]
\label{prop:chow-rashevskii}
Consider a driftless control-affine system $\dot{x} = \sum_{i=1}^m u_i f_i(x)$ on a connected
manifold $M$ with a proper control set $U$. If the family of vector fields $\{f_1, \dots, f_m\}$ is
bracket generating on $M$, then the system is small-time locally
controllable (STLC) on $M$.
\end{proposition}

\begin{proposition}[Local Controllability Implies Global Controllability]
\label{prop:local-global-controllability}
Consider a control system defined on a connected finite-dimensional
smooth manifold $M$. If the system is locally controllable, then it is
globally controllable~\cite[Theorem 1]{boscain21}.
\end{proposition}

\begin{proposition}[Krener's Theorem]
\label{prop:krener}
Consider \eqref{eq:sys_x}
and suppose the vector fields in $\mathcal{F}$ are real-analytic (in
particular, algebraic). For any $q \in M$, the attainable set $A_x(\le T, q)$
has nonempty interior in $M$ for every $T > 0$ (and hence $A_x(q)$
has nonempty interior) if and only if the system is bracket
generating at $q$, i.e., $\operatorname{Lie}_q(\mathcal{F}) = T_q M$.
\end{proposition}

\begin{remark}
Proposition~\ref{prop:chow-rashevskii} is the workhorse for our driftless
reductions. Proposition~\ref{prop:krener} is used as a Lie-rank test for the
existence of reachable directions under analytic-
ity/algebraicity. Because STLC strictly implies local controllability,
Propositions~\ref{prop:chow-rashevskii} and \ref{prop:local-global-controllability} together guarantee that an everywhere
STLC system on a connected manifold is globally controllable.
\end{remark}

We will also use state equivalence to
pass controllability between systems. Consider another control
system
\begin{equation} \label{eq:sys_y}
\Sigma_y: \quad \dot{y} = g(y, u), \quad y \in N, \quad u \in U,
\end{equation}
with the same control set $U$. We say that $\Sigma_x$ and $\Sigma_y$ are state
equivalent if there exists a diffeomorphism $\Phi : M \rightarrow N$ such
that for every admissible control $u(\cdot)$ and every solution $x(\cdot)$
of $\Sigma_x$, the curve $y(t) := \Phi(x(t))$ is a solution of $\Sigma_y$ under
the same control $u(\cdot)$.

\begin{proposition}[State equivalence preserves controllability]
\label{prop:state-equivalence}
If $\Sigma_x$ and $\Sigma_y$ are state equivalent, then $\Sigma_x$ is globally controllable if and
only if $\Sigma_y$ is globally controllable.
\end{proposition}



\section{Problem Formulation}

We formulate the steering problem on the relative interior of learners' strategy space and introduce the projected dual coordinates that will be used throughout the paper.

\subsection{Relative interior of the state space and projected mirror coordinates}

Consider $N$ learners and one controller. Learner $i$ has action set $[n_i]$, the controller has action set $[m]$, and learner $i$ uses the FTRL choice map
\[
Q_i(y_i)
=
\argmax_{x_i' \in \Delta_{n_i}}
\bigl\{
\langle x_i', y_i \rangle - h_i(x_i')
\bigr\},
\]
where $h_i$ satisfies Assumption~1. Define
\begin{equation}
\label{eq:def-n-player-HPX}
\begin{aligned}
    \cX^\circ
&:=
\Delta_{n_1}^\circ \times \cdots \times \Delta_{n_N}^\circ,\\
H_i
&:=
\{ z_i \in \mathbb{R}^{n_i} : \langle z_i, \mathbf{1}_{n_i} \rangle = 0 \},
\\
P_{H_i}
&:=
I_{n_i} - \frac{1}{n_i}\mathbf{1}_{n_i}\mathbf{1}_{n_i}^\top,\\
H &:= H_1 \times \cdots \times H_N,\\
P_H &:= \mathrm{diag}(P_{H_1}, \ldots, P_{H_N}).
\end{aligned}
\end{equation}
For convenience, write
\begin{equation}
\label{eq:def-n-player-Qh}
\begin{aligned}
    Q &:= Q_1 \oplus \cdots \oplus Q_N, \\
    \nabla h(x) &:= \nabla h_1(x_1) \oplus \cdots \oplus \nabla h_N(x_N).
\end{aligned}
\end{equation}
For each learner define
\begin{equation}
\label{eq:def-n-player-M}
\begin{aligned}
    M_i &:= P_{H_i}\nabla h_i(\Delta_{n_i}^\circ) \subset H_i,\\
M &:= M_1 \times \cdots \times M_N \subset H.
\end{aligned}
\end{equation}

Our formulation is based upon the following crucial lemma. We state this lemma as follows, but we will delay the proof to the next section.
\begin{lemma}
\label{lem:diffeomorphism}
For each $i$, the restriction
\[
\Phi_i := Q_i|_{M_i} : M_i \to \Delta_{n_i}^\circ
\]
is a diffeomorphism with inverse
\[
\Phi_i^{-1}(x_i) = P_{H_i}\nabla h_i(x_i).
\]
Consequently,
\[
\Phi := \Phi_1 \oplus \cdots \oplus \Phi_N : M \to X^\circ
\]
is a diffeomorphism with inverse
\[
\Phi^{-1}(x) = P_H \nabla h(x).
\]
\end{lemma}

Because each $Q_i$ is invariant under translations along $\mathrm{span}\{\mathbf{1}_{n_i}\}$, the joint map $Q$ satisfies $Q(y) = Q(P_H y)$. 
Moreover, for every $z \in M$,
\[
DQ(z)P_H = DQ(z),
\qquad
\mathrm{im}\,DQ(z) \subset T_{Q(z)}\cX^\circ = H.
\]
These identities allow us to pass freely between primal and projected dual coordinates, as we will make clear in the following subsections.

\begin{remark}[Why we work on $\cX^\circ$]
The controllability results in this paper are stated on the relative interior of the state space $\cX^\circ$. This is the natural domain of the projected mirror coordinates $\Phi^{-1}$, and it is forward invariant for important FTRL dynamics such as replicator dynamics. We therefore study steering on the relative interior only. We do not claim boundary inaccessibility for every regularizer covered by Assumption~1.
\end{remark}

\subsection{Two-player case}
\label{sec:two-player-formulation}

We first consider one learner with $n$ actions and one controller with $m$ actions. Let $A \in \mathbb{R}^{n \times m}$ 
be the payoff matrix mapping the controller mixed strategy $u \in \Delta_m$ to the learner payoff vector $Au$. The raw FTRL dynamics are
\begin{equation}
x = Q(y),
\qquad
\dot y = Au,
\qquad
u \in \Delta_m .
\label{eq:two-player-raw}
\end{equation}
Define
\begin{equation}
\label{eq:def-two-player-HPM}
\begin{aligned}
    H &:= \{ z \in \mathbb{R}^n : \langle z, \mathbf{1}_n \rangle = 0 \},
    \\
    P_H &:= I_n - \frac{1}{n}\mathbf{1}_n\mathbf{1}_n^\top,
    \\
    M &:= P_H \nabla h(\Delta_n^\circ),
\end{aligned}
\end{equation}
and let $\Phi := Q|_M : M \to \Delta_n^\circ$ 
be the diffeomorphism from Lemma~1. With the projected dual state $z := P_H y$, 
the system becomes
\begin{equation}
\Sigma_z^{2p}:
\qquad
\dot z = P_H A u,
\qquad
z \in M,
\quad
u \in \Delta_m .
\label{eq:two-player-dual}
\end{equation}
Transporting this dynamics to the relative interior of the probability simplex via $\Phi$ yields the induced control system
\begin{equation}
\begin{aligned}
    \Sigma_x^{2p}:
    \quad
    \dot x
    &=
    D\Phi(\Phi^{-1}(x))\,P_H A u \\
    &=
    DQ(\Phi^{-1}(x))\,A u,
    \quad
    x \in \Delta_n^\circ,
    \ 
u \in \Delta_m.
\end{aligned}
\label{eq:two-player-primal}
\end{equation}
Systems \eqref{eq:two-player-dual} and \eqref{eq:two-player-primal} are state equivalent. Throughout the paper, controllability in the two-player case means controllability of \eqref{eq:two-player-primal} on $\Delta_n^\circ$, or equivalently of \eqref{eq:two-player-dual} on $M$. 

In this paper, we identify necessary and sufficient conditions for the controllability of~\eqref{eq:two-player-primal}, i.e., when can we steer a single FTRL agent to any desired probability distribution through repeated play of the game?

\subsection{Multi-player case}
\label{sec:multi-player-formulation}

Now consider $N$ learners and one controller. Let $p_i(x_{-i},u)$ denote the payoff vector of learner $i$. Since it is linear in the controller strategy, we write
\[
p_i(x_{-i},u) = A_i(x_{-i})u,
\qquad i \in [N].
\]
Stack the learner states and dual variables as
\[
x := x_1 \oplus \cdots \oplus x_N,
\qquad
y := y_1 \oplus \cdots \oplus y_N,
\]
and define
\[
A(x)
:=
\begin{bmatrix}
A_1(x_{-1}) \\
\vdots \\
A_N(x_{-N})
\end{bmatrix}
=
\bigl[a_1(x)\ \cdots\ a_m(x)\bigr].
\]
The raw learning dynamics are
\begin{equation}
x_i = Q_i(y_i),
\qquad
\dot y_i = A_i(x_{-i})u,
\qquad
i \in [N],
\quad
u \in \Delta_m .
\label{eq:n-player-raw-block}
\end{equation}
Equivalently,
\begin{equation}
x = Q(y),
\qquad
\dot y = A(x)u,
\qquad
x \in \cX^\circ,
\quad
u \in \Delta_m .
\label{eq:n-player-raw}
\end{equation}

With the projected dual state $z := P_H y$ and the diffeomorphism $\Phi : M \to \cX^\circ$, the state-equivalent projected dual dynamics are
\begin{equation}
\Sigma_z:
\qquad
\dot z = P_H A(\Phi(z))u,
\qquad
z \in M,
\quad
u \in \Delta_m .
\label{eq:n-player-dual}
\end{equation}
The induced primal dynamics on the product of simplices are
\begin{equation}
\begin{aligned}
        \Sigma_x:
    \quad
    \dot x
    &=
    D\Phi(\Phi^{-1}(x))\,P_H A(x)u \\
    &=
    DQ(\Phi^{-1}(x))\,A(x)u,
    \quad
    x \in \cX^\circ,
    \ 
    u \in \Delta_m.
\end{aligned}
\label{eq:n-player-primal}
\end{equation}
Again, \eqref{eq:n-player-dual} and \eqref{eq:n-player-primal} are state equivalent. All controllability statements for the multi-learner model refer to \eqref{eq:n-player-primal}, while the proofs will often be carried out in the projected dual coordinates \eqref{eq:n-player-dual}. 

In this paper, we identify multiple sufficient conditions for the controllability of \eqref{eq:n-player-primal}, i.e., when can we simultaneously steer multiple FTRL agents to desired states through repeated play of the game?

\section{Main Results}\label{sec:main}
\subsection{The 2-Player Case}
We first analyze the setting of a single controller interacting with a single learner. Building on the dynamics established in the problem formulation, our goal is to determine exactly when the controller can steer the learner to any target mixed strategy within the relative interior of the simplex. A key technical challenge in analyzing this setup is the learner's continuous, state-dependent drift. To address this, we explore how the controller can effectively cancel out this autonomous drift. Doing so allows us to transform the original nonlinear dynamics into a simpler driftless control system, which ultimately reduces the complex question of steerability to a straightforward algebraic rank test.

Our first result gives a necessary and sufficient condition for controllability of \eqref{eq:two-player-primal} on $\Delta_{n}^{\circ}$. The condition is based on the concept of a neutralizing strategy: a mixed strategy such that, if the controller plays it, the learner’s payoff vector is constant (and hence independent of the learner’s current action). Formally:
\begin{definition}[Neutralizing Strategy]\label{def:neutralizing}
    A mixed strategy $u_0\in\Delta_m$ is neutralizing if $A u_0 = k\,\mathbf 1$ with some $k\in\mathbb{R}$.
\end{definition}
Intuitively, when the controller plays a neutralizing strategy, the learner is made indifferent among its actions because all actions yield the same expected payoff. In the controllability analysis below, the existence of a fully mixed neutralizing strategy ensures that the projected control directions contain a neighborhood of the origin, which is essential for moving the learner’s state in all directions within the relative interior.

Now, let $H, P_H$ and $M$ be defined as in \eqref{eq:def-two-player-HPM}. Our first result can be stated as follows:
\begin{theorem}[Controllability of the two-player system]
    \label{thm:two-player}
    The system \eqref{eq:two-player-primal} is controllable on $\Delta_{n}^{\circ}$ if and only if
    \begin{enumerate}
        \item There exists a fully mixed neutralizing strategy $u_{0} \in \Delta_{m}^{\circ}$ for the controller, and
        \item $\rank(P_{H} A) = n-1$.
    \end{enumerate}
\end{theorem}
As discussed briefly in Section~\ref{sec:two-player-formulation}, our main proof technique is to transform system \eqref{eq:two-player-primal} into a state-equivalent system with a simpler form, i.e.
\begin{equation}
    \label{eq:two-player-auxiliary}
    z = P_{H} A u, \quad z \in P_{H}\nabla h(\Delta_{n}^{\circ}), \quad u \in U.
\end{equation}
The state-equivalence between the primal and dual systems is a consequence of Lemma~\ref{lem:diffeomorphism}, which we stated without proof. We now state the complete two-player version of this lemma as follows, and give a full proof of it.
\begin{lemma}
    \label{lem:two-player-diffeomorphism}
    The restriction of the choice map
\[
\Phi := Q|_{M} : M \to \Delta^\circ_n
\]
is a diffeomorphism with inverse
\[
\Phi^{-1}(x) = P_{H}\nabla h(x).
\]
\end{lemma}
\begin{proof}
    We first confirm that $Q$ is well defined and continuous. For fixed $y$, the function $x\mapsto\langle x,y\rangle-h(x)$ is strictly concave. Since $\Delta_n$ is compact and convex, the maximizer $Q(y)$ exists and is unique. Continuity of $Q$ as a function of $y$ follows from Berge’s maximum theorem~\cite{ok07}.

    To prove the lemma, it suffices to show that (i) $Q|_{M}: M \to \Delta_{n}^{\circ}$ is bijective, and (ii) $Q|_{M}$ is smooth.

    To show bijectivity, we examine the optimization problem whose solution is $Q(y)$:
    \begin{equation}
        \label{eq:optmization}
        \begin{aligned}
            \minimize_{z \in \Delta_{n}}\, &h(z) - \langle z, y \rangle \\
            s.t.\, &\langle \mathbf{1}, z \rangle - 1 = 0, \\
                   &z_{i} \ge 0.
        \end{aligned}
    \end{equation}
    The Lagrangian is
    \begin{equation*}
        \label{eq:Lagrangian}
        \cL(z, \lambda, \mu)
        =
        h(z) - \langle z, y \rangle + \lambda(\langle \mathbf{1}, z \rangle - 1) - \langle \mu, z \rangle,
    \end{equation*}
    where $\lambda \in \bbR$ and $\mu \in \bbR^{n}_{\ge 0}$. By the KKT conditions, the solution to \eqref{eq:optmization}, $x = Q(y)$, must satisfy
    \begin{equation*}
        \begin{cases}
            \nabla_{z} \cL(z, \lambda, \mu) = \nabla h(x) - y + \lambda \mathbf{1} - \mu = 0,\\
            \langle 1, x \rangle - 1 = 0,\\
            \mu_{i} \ge 0, \, \mu_{i} x_{i} = 0 \text{ for } i \in [n].\\
        \end{cases}
    \end{equation*}
    Since $x \in \Delta_{n}^{\circ}$, these conditions reduce to
    \begin{equation}
        \label{eq:KKT}
        \nabla h(x) - y + \lambda \mathbf{1} = 0,
    \end{equation}
    which is also sufficient for optimality. In fact, this also gives us a exact characterization of the preimage $Q^{-1}(\Delta_{n}^{\circ})$:
    \begin{equation}
        \label{eq:Q-preimage}
        Q^{-1}(\Delta_{n}^{\circ}) = \nabla h(\Delta_{n}^{\circ}) + \laspan\{\mathbf{1}\}.
    \end{equation}
    Further, since $\langle x, \mathbf{1} \rangle$ is constant over $\Delta_n$, we have $Q(y) = Q(y + k\mathbf{1})$ for any $y \in \bbR^{n}$ and $k \in \bbR$. This implies that
    \begin{equation*}
        Q(\nabla h(\Delta_{n}^{\circ}) + \laspan\{\mathbf{1}\}) = Q(P_{H} \nabla h(\Delta_{n}^{\circ})) = \Delta_{n}^{\circ},
    \end{equation*}
    so $Q|_{P_{H}\nabla h(\Delta_{n}^{\circ})}$ is surjective onto $\Delta_{n}^{\circ}$. To show injectivity, take any $y_{1}, y_{2} \in P_{H}\nabla h(\Delta_{n}^{\circ})$. If $Q(y_{1}) = Q(y_{2}) = x$, then \eqref{eq:KKT} yields
    \begin{equation*}
        y_{1} - y_{2} = (\lambda_{1} - \lambda_{2})\mathbf{1},
    \end{equation*}
    which can only be zero since $y_{1}, y_{2} \in H$. This establishes bijectivity.

    It remains to show smoothness. Define the map
    \begin{equation*}
        F(x, \lambda; y) := 
        \begin{bmatrix}
            \nabla h(x) - y + \lambda \mathbf{1} \\
            \langle 1, x \rangle - 1 
        \end{bmatrix}.
    \end{equation*}
    Since $h$ is smooth, $F$ is smooth, and $F(x, \lambda; y) = 0$ corresponds exactly to the KKT condition. The Jacobian of $F$ is
    \begin{equation*}
        J = \begin{bmatrix}
            \nabla^{2}h(x) &\mathbf{1}\\
            \mathbf{1}^{\T} &0
        \end{bmatrix},
    \end{equation*}
    which is invertible since $\nabla^{2}h$ is invertible. By the Implicit Function Theorem, for every $(x_{0}, \lambda_{0}, y_{0})$ with $F(x_{0}, \lambda_{0}; y_{0}) = 0$, there exist neighborhoods of $y_{0}$ and $(x_{0}, \lambda_{0})$ and a unique smooth mapping $(x(\cdot), \lambda(\cdot))$ such that $F(x(y), \lambda(y); y) = 0$ in that neighborhood. When $x \in \Delta_{n}^{\circ}$, this solution coincides with the unique maximizer of \eqref{eq:optmization}, i.e., $x(y)=Q(y)$. Since this holds for every $y \in Q^{-1}(\Delta_{n}^{\circ})$, we conclude that $Q$ is smooth on $Q^{-1}(\Delta_{n}^{\circ}) = \nabla h(\Delta_{n}^{\circ}) + \laspan\{\mathbf{1}\}$. Since $M = P_{H}\nabla h(\Delta_{n}^{\circ}) \subset \nabla h(\Delta_{n}^{\circ}) + \laspan\{\mathbf{1}\}$, the claim follows.
\end{proof}

We are now ready to prove Theorem~\ref{thm:two-player}.

    \begin{proof}[Proof of Theorem~\ref{thm:two-player}]
    As discussed in Section~\ref{sec:two-player-formulation}, the primal system \eqref{eq:two-player-primal} is state equivalent to the projected dual system \eqref{eq:two-player-dual}; it therefore suffices to analyze controllability on $M$.
    
    \noindent\emph{Step 1: Necessity of 1).}
    
    Assume that there is no fully mixed neutralizing strategy. Since $P_HA$ is linear and $\Delta_m^\circ$ is the relative interior of $\Delta_m$, we have $P_HA(\Delta_m^\circ)=\big(P_HA(\Delta_m)\big)^\circ$. 
    Therefore $0\notin \big(P_HA(\Delta_m)\big)^\circ$. By the separating hyperplane theorem, there exists a nonzero vector $w\in H$ such that $\langle w,P_HAu\rangle\ge 0$ for all $u\in\Delta_m$. 
    Let $z(\cdot)$ be any trajectory of \eqref{eq:two-player-dual}. Then
    $\frac{d}{dt}\langle w,z(t)\rangle
    =
    \langle w,P_HAu(t)\rangle
    \ge 0$, 
    so the scalar function $t\mapsto \langle w,z(t)\rangle$ is non-decreasing along every trajectory.
    
    Fix any $z_0\in M$. Since $w\in H=T_{z_0}M$, there exists a smooth curve
    $\gamma:(-\varepsilon,\varepsilon)\to M$ 
    such that $\gamma(0)=z_0$ and $\dot\gamma(0)=-w$. Consequently, 
    $\langle w,\gamma(s)\rangle<\langle w,z_0\rangle$ 
    for all sufficiently small $s>0$. Such points cannot be attained from $z_0$, which shows that \eqref{eq:two-player-dual} is not locally controllable, and hence not globally controllable. Thus 1) is necessary.

    Assume now that 1) holds.
    Choose $u_0\in\Delta_m^\circ$ such that $P_HAu_0=0$, and define
    $V:=\Delta_m-u_0$. 
    Since $u_0\in\Delta_m^\circ$, the set $V$ is a proper control set in the affine subspace
    $\operatorname{aff}(V)=\{v\in\mathbb R^m:\langle v,1_m\rangle=0\}$, 
    and $0\in V^\circ$. Writing $u=u_0+v$, the auxiliary system becomes
    \begin{equation}
        \label{eq:two-player-auxiliary-3}
        \dot z=P_HAv,\qquad z\in M,\quad v\in V.
    \end{equation}
    
    \noindent\emph{Step 2: Necessity of 2).}
    
    Let $z(\cdot)$ be any trajectory of \eqref{eq:two-player-auxiliary-3} with initial condition $z(0)=z_0$. Then for every $T\ge 0$,
    \[
    z(T)-z_0=\int_0^T P_HAv(t)\,dt\in \operatorname{im}(P_HA).
    \]
    Hence the attainable set from $z_0$ satisfies $A_z(z_0)\subset \big(z_0+\operatorname{im}(P_HA)\big)\cap M$. 
    If $\operatorname{rank}(P_HA)<n-1$, then $z_0+\operatorname{im}(P_HA)$ is a proper affine subspace of $H$. Since $M$ is an $(n-1)$-dimensional manifold, it cannot be contained in a proper affine subspace of $H$. Therefore $A_z(z_0)\neq M$, 
    so \eqref{eq:two-player-auxiliary-3} is not controllable. Thus 2) is necessary.
    
    \noindent\emph{Step 3: Sufficiency of 2).}
    
    Write the columns of $P_HA$ as $P_HA=[b_1\ \cdots\ b_m],\qquad b_j\in H$. 
    Then \eqref{eq:two-player-auxiliary-3} is the driftless control-affine system $\dot z=\sum_{j=1}^m v_j b_j$, with $v\in V$, 
    whose control vector fields are constant on $M$. Therefore all Lie brackets vanish, and for every $z\in M$,
    \[
    \operatorname{Lie}_z\{b_1,\dots,b_m\}
    =
    \operatorname{span}\{b_1,\dots,b_m\}
    =
    \operatorname{im}(P_HA).
    \]
    If $\operatorname{rank}(P_HA)=n-1$, then
    \[
    \operatorname{Lie}_z\{b_1,\dots,b_m\}
    =
    H
    =
    T_zM,
    \qquad \forall z\in M.
    \]
    Thus the family $\{b_1,\dots,b_m\}$ is bracket generating on $M$. Since $V$ is a proper control set, Proposition~\ref{prop:chow-rashevskii} implies that \eqref{eq:two-player-auxiliary-3} is small-time locally controllable everywhere on $M$. By Proposition~\ref{prop:local-global-controllability}, it follows that \eqref{eq:two-player-auxiliary-3} is globally controllable on the connected manifold $M$. Finally, Proposition~\ref{prop:state-equivalence} transfers controllability back to the original system, so the primal system is controllable on $\Delta_n^\circ$. 
    This completes the proof.
    \end{proof}

\subsection{The N-player case}
Moving beyond a single opponent, we now consider the case where the controller interacts with multiple independent learners at the same time. Steering a joint system is naturally more demanding, as any change in the controller's strategy simultaneously affects the entire population. Because finding a tight necessary and sufficient condition is difficult in this high-dimensional setting, we instead focus on providing reliable sufficient conditions. We outline two distinct approaches. The first generalizes our earlier intuition by asking if the controller can uniformly neutralize the drift for all learners at once. Since this can be a restrictive requirement in larger games, we also provide a complementary condition that relies on the periodicity of the system's drift.

We will start by laying out the first sufficient condition, which is based on the concept of a uniformly neutralizing strategy:
\begin{definition}[Uniformly Neutralizing Strategy]
    A mixed strategy $u_0\in\Delta_m$ is uniformly neutralizing if for every fully mixed strategy profile $x$, $A_{i}(x) u_0 = k_{i}\,\mathbf 1$ with some $k_{i} \in \bbR$ for every $i \in [N]$.
\end{definition}
This extends the two-player notion of a neutralizing strategy: a uniformly neutralizing strategy simultaneously neutralizes each learner’s payoff vector across all system states. As a result, it removes the drift term in the projected auxiliary system, reducing the controllability analysis to a driftless control-affine problem.

Let $H_i, H, P_H$ and $\cX^\circ$ be defined as in \eqref{eq:def-n-player-HPX},
let $Q$ and $h$ be defined as in \eqref{eq:def-n-player-Qh}, 
and let $M$ be defined as in \eqref{eq:def-n-player-M}.
Define the vector fields $\eta_1, \dots, \eta_{m-1} : \cX^\circ \to \cX^\circ$ by
\begin{equation}
    \label{eq:z-fields}
    \eta_i (x) = DQ(\nabla h(x))\left(a_i(x) - a_m(x) \right). 
\end{equation}
\begin{remark} 
\label{rmk:nabla-h}
    Since $\Phi^{-1}(x)=P_H\nabla h(x)$ and $Q$ is invariant under blockwise
    translations along $H^\perp$, we have
    \[
    DQ(\Phi^{-1}(x))=DQ(P_H \nabla h(x))= DQ(\nabla h(x)).
    \]
    Accordingly, in theorem statements we use
    \[
    \eta_i(x):=DQ(\nabla h(x))(a_i(x)-a_m(x)), \  i=1,\dots,m-1.
    \]
\end{remark}

For a family $\cG$ of smooth vector fields on $X^\circ$, define
\[
\lie_x(\cG):=\{Y(x):Y\in\lie(\cG)\}\subset T_xX^\circ.
\]
We say that $\lie(G)$ has full rank on $\cX^\circ$ if
$\lie_x(\cG)=T_x\cX^\circ$ for all $x\in \cX^\circ$, or, equivalently, $\dim \lie_x(\cG)=\dim(\cX^\circ)$ for all $x\in \cX^\circ$. 
Our sufficient condition for controllability of the $N$-learner system is as follows:
\begin{theorem}[Sufficient condition for controllability of the $N$-player system with a uniformly neutralizing strategy]
    \label{thm:N-player-controllability-1}
    The system \eqref{eq:n-player-primal} is controllable if 
    \begin{enumerate}
        \item There exists a fully mixed uniformly neutralizing strategy $u_{0} \in \Delta_{m}^{\circ}$ for the controller, and
        \item $\lie\{\eta_1,\ldots,\eta_{m-1}\}$ has full rank on $X^\circ$.
    \end{enumerate}
\end{theorem}

As in the two-player case, we prove this by analyzing the dual system granted by Lemma~\ref{lem:diffeomorphism}. We have proven the two-player version of Lemma~\ref{lem:diffeomorphism} in the previous section, now we state and prove the second part, i.e. the multi-player version.
\begin{lemma}
    \label{lem:n-player-diffeomorphism}
    The restriction of the choice map
\[
\Phi := Q|_M = \Phi_1 \oplus \cdots \oplus \Phi_N : M \to \cX^\circ
\]
is a diffeomorphism with inverse
\[
\Phi^{-1}(x) = P_H \nabla h(x).
\]
\end{lemma}
\begin{proof}
Apply the single-learner argument to each block. The KKT conditions show that $Q_i^{-1}(\Delta_{n_i}^\circ)
=
\nabla h_i(\Delta_{n_i}^\circ) + \mathrm{span}\{\mathbf{1}_{n_i}\}$, 
and restriction to $H_i$ removes the nonuniqueness along the all-ones direction. Smoothness follows from the implicit function theorem exactly as in the one-learner case.
\end{proof}

Now we are ready to prove Theorem~\ref{thm:N-player-controllability-1}.
\begin{proof}[Proof of Theorem~\ref{thm:N-player-controllability-1}]\
Let $d:=\sum_{i=1}^N (n_i-1)=\sum_{i=1}^N n_i-N$. 
As in the two-player case, the primal system \eqref{eq:n-player-primal} and the dual system \eqref{eq:n-player-dual} are state-equivalent due to Lemma~\ref{lem:n-player-diffeomorphism}. We will focus on the analysis of the dual system for simplicity. Recall that the dual system is given by
\[
\Sigma_z:\quad
\dot z=P_HA(\Phi(z))u,\quad z\in M,\  u\in\Delta_m,
\]
and the primal dynamics is given by
\[
\Sigma_x:\quad
\dot x=DQ(\Phi^{-1}(x))A(x)u,\quad x\in \cX^\circ,\  u\in\Delta_m.
\]
Moreover, $\Sigma_z$ and $\Sigma_x$ are state equivalent via the diffeomorphism
$\Phi$. Indeed, for $x=\Phi(z)$,
\[
D\Phi(z)\big(P_HA(\Phi(z))u\big)
=
DQ(z)P_HA(x)u
=
DQ(z)A(x)u,
\]
where we again used $DQ(z)P_H=DQ(z)$.
Therefore, by Proposition~\ref{prop:state-equivalence}, controllability of $\Sigma_x$ on $\cX^\circ$ is equivalent
to controllability of $\Sigma_z$ on $M$.

\noindent\emph{Step 1: Neutralizing the drift and introducing effective controls.} 
By Condition 1, there exists $u_0\in\Delta_m^\circ$ such that for every
$x\in \cX^\circ$ and every learner $i$, $A_i(x)u_0=k_i(x)\mathbf 1_{n_i}$ 
for some scalar $k_i(x)\in\mathbb R$.
Applying $P_{H_i}$ to each block gives $P_{H_i}A_i(x)u_0=0$, 
and hence, after stacking the blocks, $P_HA(x)u_0=0$ for all $x\in \cX^\circ$.

Introduce the matrix
\[
E:=
\begin{bmatrix}
e_1-e_m & \cdots & e_{m-1}-e_m
\end{bmatrix}
\in\mathbb R^{m\times (m-1)}.
\]
Since $\operatorname{im}(E)=\{v\in\mathbb R^m:\langle v,\mathbf 1_m\rangle=0\}$,
every $u\in\Delta_m$ can be written uniquely as $u=u_0+Ew$ for some $w\in\mathbb R^{m-1}$.
Define $W:=\{w\in\mathbb R^{m-1}:u_0+Ew\in\Delta_m\}$. 
Because $u_0\in\Delta_m^\circ$, the translated simplex $\Delta_m-u_0$ is a proper
control set in $\{v:\langle v,\mathbf 1_m\rangle=0\}$, and since $E$ is a linear
isomorphism from $\mathbb R^{m-1}$ onto this hyperplane, $W$ is a proper control
set in $\mathbb R^{m-1}$.

Substituting $u=u_0+Ew$ into $\Sigma_z$ yields
\begin{equation*}
\begin{aligned}
    \dot z
&=
P_HA(\Phi(z))(u_0+Ew)\\
&=
P_HA(\Phi(z))Ew \\
&=
\sum_{k=1}^{m-1} w_k\, b_k(z),
\quad w\in W,
\end{aligned}  
\end{equation*}
where
\[
b_k(z):=P_H\big(a_k(\Phi(z))-a_m(\Phi(z))\big),
\quad k=1,\dots,m-1.
\]
Thus the projected dual system is a driftless control-affine system on $M$ with
control vector fields $b_1,\dots,b_{m-1}$.

\noindent\emph{Step 2: Push-forward of the control vector fields.} 
For $x=\Phi(z)$,
\[
(\Phi_*b_k)(x)
=
D\Phi(\Phi^{-1}(x))\,b_k(\Phi^{-1}(x)).
\]
Since $b_k(z)=P_H\bigl(a_k(\Phi(z))-a_m(\Phi(z))\bigr)$, 
we obtain
$(\Phi_*b_k)(x)
=
D\Phi(\Phi^{-1}(x))\,P_H\bigl(a_k(x)-a_m(x)\bigr)$. 
Using Remark~\ref{rmk:nabla-h} and the identity $DQ(\zeta)P_H=DQ(\zeta)$, this becomes
\[
(\Phi_*b_k)(x)
=
DQ(\nabla h(x))\bigl(a_k(x)-a_m(x)\bigr)
=
\eta_k(x).
\]
Therefore,
\[
\Phi_*\bigl(\lie\{b_1,\ldots,b_{m-1}\}\bigr)
=
\lie\{\eta_1,\ldots,\eta_{m-1}\}.
\]
Hence, for every $z\in M$ and $x=\Phi(z)$,
\[
D\Phi(z)\,\lie_z\{b_1,\ldots,b_{m-1}\}
=
\lie_x\{\eta_1,\ldots,\eta_{m-1}\}.
\]

\noindent\emph{Step 3: Wrapping up the proof.} 
By Condition 2), $\lie\{\eta_1,\ldots,\eta_{m-1}\}$ has full rank on $\cX^\circ$, i.e., $\lie_x\{\eta_1,\ldots,\eta_{m-1}\}=T_x\cX^\circ$ for all $x \in \cX^\circ$. 
Since $D\Phi(z):T_zM\to T_xX^\circ$ is a linear isomorphism, it follows that $\lie_z\{b_1,\ldots,b_{m-1}\}=T_zM$ for all $z \in M$. 
Thus the family $\{b_1,\ldots,b_{m-1}\}$ is bracket generating on $M$.

Since $W$ is a proper control set, Proposition~\ref{prop:chow-rashevskii} implies that the driftless dual
system is small-time locally controllable everywhere on $M$. The manifold $M$ is connected
because it is diffeomorphic to the connected manifold $\cX^\circ$. Hence, by
Proposition~\ref{prop:local-global-controllability}, the dual system is globally controllable on $M$.
Finally, Proposition~\ref{prop:state-equivalence} transfers controllability back through the diffeomorphism
$\Phi$, so the primal system \eqref{eq:n-player-primal} is controllable on $\cX^\circ$. 
This completes the proof.
\end{proof}

For a concrete example, please see Section~\ref{ex:BI}.

Intuitively, the first condition of Theorem~\ref{thm:N-player-controllability-1} ensures that the auxiliary system has no drift term. However, controllability is sometimes possible even with drift. We provide below another sufficient condition that captures one such case, which is due to~\cite[Chapter~4, Theorem~5]{jurdjevic96}, which we present here as Proposition~\ref{prop:recurrent}. This condition relies on periodicity.
\begin{definition}[Periodicity of Vector Fields]
    Let $f$ be a vector field on $M$. We say $f$ is periodic if $t \mapsto e^{tf}(x)$ is periodic for every $x \in M$.
\end{definition}

\begin{proposition}{\cite[Chapter~4, Theorem~5]{jurdjevic96}}
\label{prop:recurrent}
    Suppose that $X_0,\ldots,X_m$ are vector fields on $M$ that define a system affine in control. Assume that
    \begin{enumerate}
        \item[(a)] the drift $X_0$ is periodic, and
        \item[(b)] the origin of $\mathbb{R}^m$ lies in the interior of the convex hull of $U$.
    \end{enumerate}
    Then the corresponding control system is controllable provided that $\operatorname{Lie}_x\{X_0,\ldots,X_m\} = T_x M$ for each $x \in M$.
\end{proposition}

Further, let $\bar u:=\frac1m\mathbf{1}_m$. 
Under the change of variables $u=\bar u+Ew$, the induced primal system
\eqref{eq:n-player-primal} can be written as
\begin{equation*}
\dot x=\eta_0(x)+\sum_{k=1}^{m-1} w_k\eta_k(x),
\qquad
w\in W,
\end{equation*}
where
\begin{equation}
\eta_0(x):=\frac1m DQ(\nabla h(x))A(x)\mathbf{1}_m.
\label{eq:eta0-def}
\end{equation}

Then our second sufficient condition can be stated as:
\begin{theorem}[Sufficient condition for controllability of the $N$-player system with a periodic drift]
    \label{thm:N-player-controllability-2}
    The system \eqref{eq:n-player-primal} is controllable if 
    \begin{enumerate}
        \item the vector field $\eta_0$ defined in \eqref{eq:eta0-def} is periodic on $\cX^\circ$, and
        \item $\lie\{\eta_1,\ldots,\eta_{m-1}\}$ has full rank on $X^\circ$.
    \end{enumerate}
    where $\eta_1, \dots, \eta_{m-1}$ are as defined in \eqref{eq:z-fields}.
\end{theorem}
\begin{proof}[Proof of Theorem~\ref{thm:N-player-controllability-2}]
Under the affine change of variables $u=\bar u+Ew$, the induced primal system
\eqref{eq:n-player-primal} becomes
$\dot x=\eta_0(x)+\sum_{k=1}^{m-1} w_k\eta_k(x)$, with $w\in W$. 
Since $\bar u=\frac1m\mathbf 1_m\in\Delta_m^\circ$ and $E:\mathbb{R}^{m-1}\to\{v\in\mathbb{R}^m:\langle v,\mathbf 1_m\rangle=0\}$ is a linear isomorphism, the set $W:=\{w\in\mathbb{R}^{m-1}:\bar u+Ew\in\Delta_m\}$ is convex, compact, and contains $0$ in its interior.

By Condition 2), the family $\{\eta_1,\ldots,\eta_{m-1}\}$ has full rank on $\cX^\circ$.
Hence the larger family $\{\eta_0,\eta_1,\ldots,\eta_{m-1}\}$ is bracket generating on $\cX^\circ$.

By Condition 1), every trajectory of the drift system $\dot x=\eta_0(x)$ is periodic; in particular,
$\eta_0$ is Poisson stable.

Therefore, we can invoke~\cite[Chapter~4, Theorem~5]{jurdjevic96}, a classical controllability theorem for control-affine systems with
Poisson-stable drift, bracket-generating vector fields, and control set whose convex hull is a neighborhood
of the origin, and conclude global controllability on $\cX^\circ$.

Since the change of variables $u=\bar u+Ew$ is bijective, the original induced primal system
\eqref{eq:n-player-primal} is controllable on $\cX^\circ$.
\end{proof}
\begin{remark}
    A sharper treatment for constrained controls with $0$ possibly on the boundary of the control set is available in~\cite{caillau25}, but we do not use it here.
\end{remark}

For a concrete example, please see Section~\ref{ex:RMP}.

\section{Examples}
\subsection{Examples for The two-player case}
\subsubsection{Rock-Paper-Scissor Game}
\label{exp:RPS}
A canonical steerable two player finite game is the Rock-Paper-Scissor (RPS) game, which is a zero sum game with the following payoff matrix:
\begin{equation*}
    A
    =
    \begin{bmatrix}
        \epsilon &-1 &1\\
        1 &\epsilon &-1\\
        -1 &1 &\epsilon
    \end{bmatrix}
\end{equation*}
with some $\epsilon \in (-1, 1)$. An obvious choice for a neutralizing strategy is $u_{0} = (1/3, 1/3, 1/3)$. Further, a straightforward computation yields
\begin{equation*}
    P_{H}A
    =
    \begin{bmatrix}
        \frac{2\epsilon}{3} &-1-\frac{\epsilon}{3} &1-\frac{\epsilon}{3}\\
        1-\frac{\epsilon}{3} &\frac{2\epsilon}{3} &-1-\frac{\epsilon}{3}\\
        -1-\frac{\epsilon}{3} &1-\frac{\epsilon}{3} &\frac{2\epsilon}{3}\\
    \end{bmatrix},
\end{equation*}
which has rank $2$ for $-1 \le \epsilon \le 1$. By Theorem~\ref{thm:two-player}, this game is steerable.
\subsubsection{Non-Steerable Case: Modified RPS}
We can break steerability with a small modification of the payoff matrix. Consider
\begin{equation*}
    A
    =
    \begin{bmatrix}
        0 &-1 &1\\
        1 &0 &-1\\
        -1 &1 &3
    \end{bmatrix}.
\end{equation*}
Despite
\begin{equation*}
    P_{H}A
    =
    \begin{bmatrix}
        0 &-1 &0\\
        1 &0 &-2\\
        -1 &1 &2
    \end{bmatrix}
\end{equation*}
which has rank $\larank(P_{H}A) = 2$, the only neutralizing strategy $u_{0} = (2/3, 0, 1/3)$ lies on the boundary of $\Delta_{3}$. Therefore, the game is not steerable. Figure~\ref{fig:attainable} illustrates the attainable set of this system from three random initial points, which does not cover $\Delta_{3}^{\circ}$.
\begin{figure}
    \center
    \resizebox{\columnwidth}{!}{%
        \includegraphics{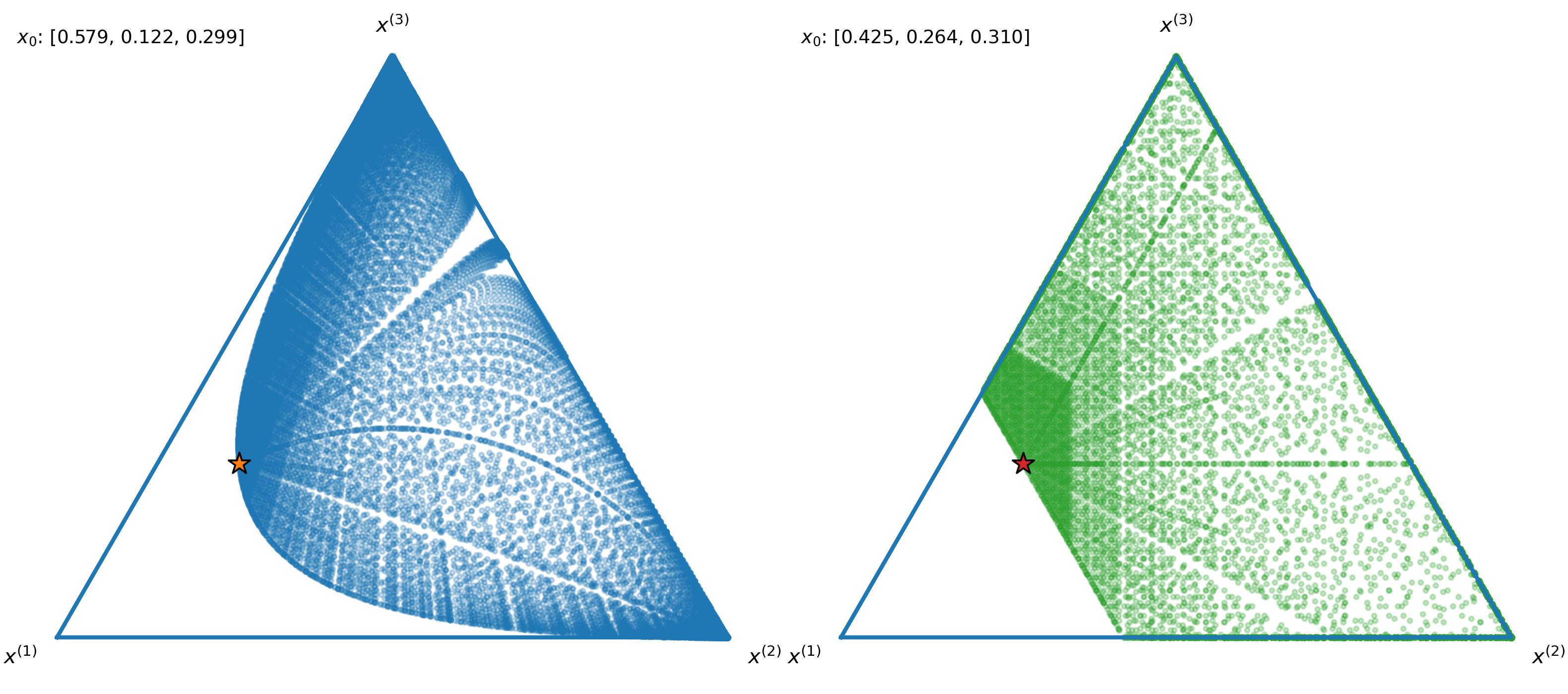}
    }
    \resizebox{\columnwidth}{!}{%
        \includegraphics{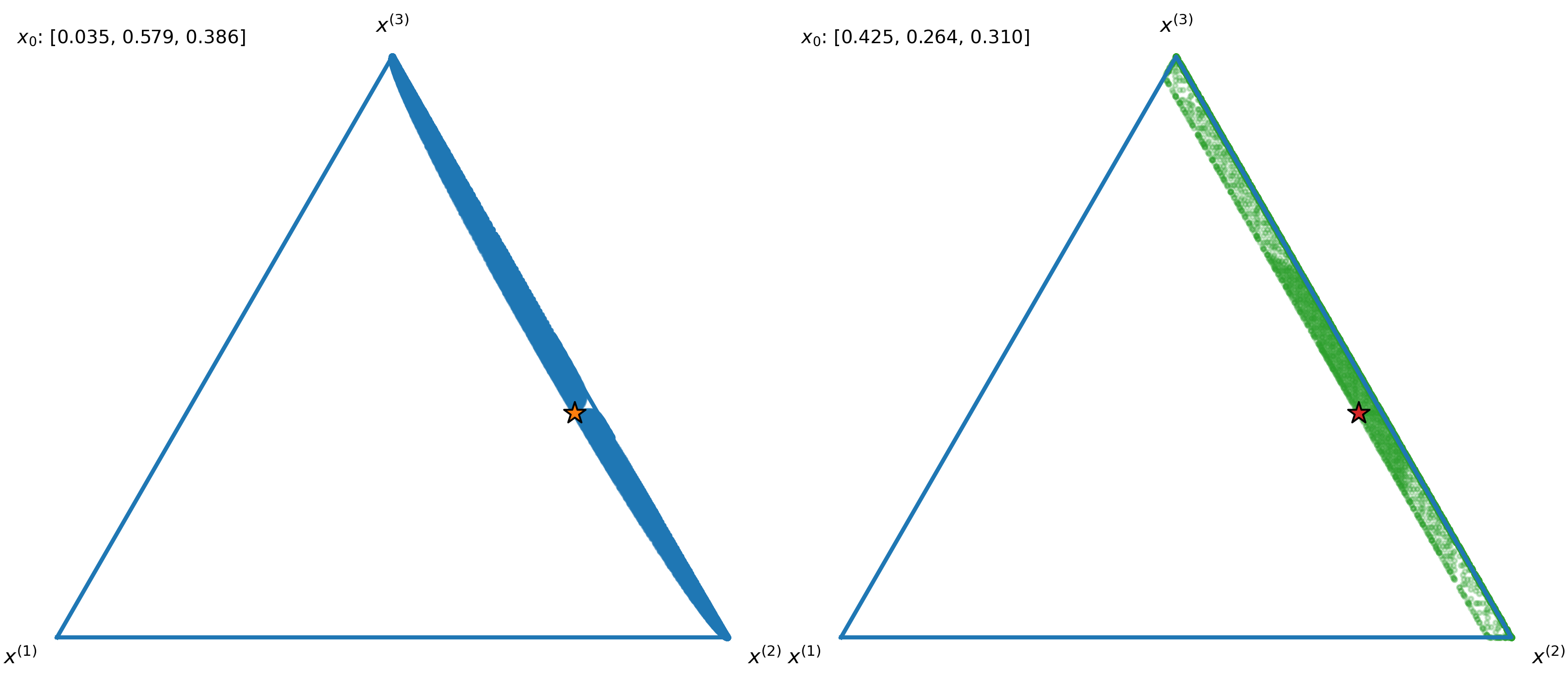}
    }
    \resizebox{\columnwidth}{!}{%
        \includegraphics{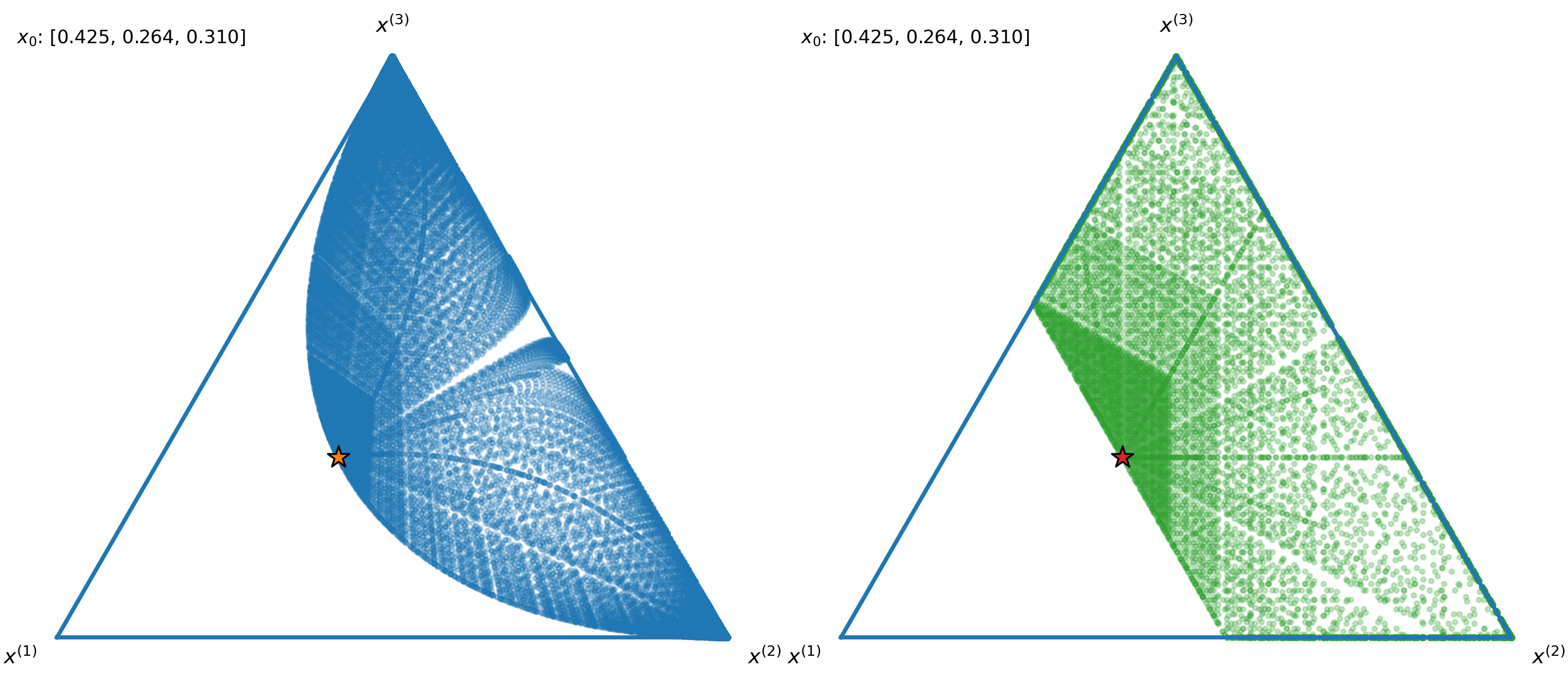}
    }
    \caption{Approximate attainable sets for the Modified RPS game: We plot the approximate attainable set for three random initial points in $\Delta_{3}^{\circ}$, for two different variants of FTRL. The figures on the left shows the case where the learner adopts the Replicator Dynamics, while the figures on the right shows that of the learner who adopts FTRL with $h(x) = 1/2\|x\|^{2}$. The reachable set is approximated by plotting the union of states generated by constant controls $u \in \Delta_3$ sampled on the simplex lattice $u=(i,j,k)/50$, $i+j+k=50$, and time horizons $t \in [0,12]$ sampled on a uniform grid of $45$ points.}
    \label{fig:attainable}
\end{figure}

\subsection{Examples for The multi-player case}
\subsubsection{Brockett's Integrator Game}
\label{ex:BI}
Consider the game with 3 learners and a controller. Each learner has 2 pure strategies, and the controller has 3 pure strategies. Learners 1 and 2 interact exclusively with the controller, receiving payoff depending on the first two and last two strategies of the controller, respectively. Learner 3 receives payoff depending on all three strategies of the controller, but the exact payoff it receives depends on the strategies of the other two learners. Figure~\ref{fig:BI} illustrates the structure of this game.
\begin{figure}[t]
    \centering
    \resizebox{\columnwidth}{!}{%
        \begin{tikzpicture}[
            >=Stealth,
            font=\small,
            node distance=2.2cm and 2.4cm,
            player/.style={draw, rounded corners, align=center, minimum width=3.0cm, minimum height=1.0cm},
            lab/.style={font=\scriptsize, inner sep=1pt}
            ]
            \node[player] (C)  {Controller\\ $u\in\Delta^3$};
            \node[player, below left=1.8cm and 2.4cm of C] (L1) {Learner 1\\ $x_1\in\Delta^2$\\$p_{1}(u^{(1)}, u^{(2)})$};
            \node[player, below=1.8cm of C] (L2) {Learner 3\\ $x_3\in\Delta^2$\\$p_{3}(x_{1}^{(1)}, x_{2}^{(1)}, u)$};
            \node[player, below right=1.8cm and 2.4cm of C] (L3) {Learner 2\\ $x_2\in\Delta^2$\\$p_{2}(u^{(2)}, u^{(3)})$};

            \draw[->] (C) -- node[lab, left]{$u^{(1)}, u^{(2)}$} (L1);
            \draw[->] (C) -- node[lab, right]{$u^{(1)}, u^{(2)}, u^{(3)}$} (L2);
            \draw[->] (C) -- node[lab, right]{$\quad u^{(2)}, u^{(3)}$} (L3);

            \draw[->] (L1) to[lab, left] node[lab, below]{$x_1^{(1)} \rightarrow r_2$} (L2);
            \draw[->] (L3) to[lab, right] node[lab, below]{$x_2^{(1)} \rightarrow r_2$} (L2);
    \end{tikzpicture} }
    \caption{Interdependency graph for the Brockett's Integrator game: This game shows the interdependency of players in the Brockett's Integrator game presented in Section~\ref{ex:BI}. As shown in the graph, learner 1's payoff only depend on the probability of the controller playing his first two strategies, and learner 2's payoff only s on the probability of the controller playing his last two strategies. Learner three's payoff depends on the probability of learner 1 and 2 playing their first strategies, as well as the entire mixed strategy of the controller.}
    \label{fig:BI}
\end{figure}
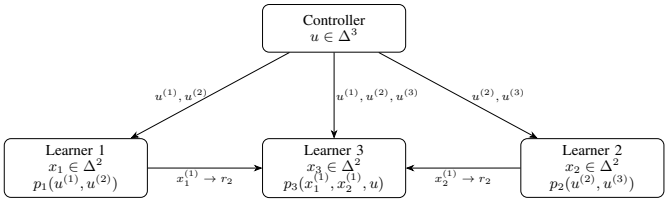
We can model this game in the form of \eqref{eq:n-player-raw-block} with
\begin{equation*} 
    \begin{cases} 
        A_{1}(x_{2}, x_{3})
        = 
        \begin{bmatrix} 
            1 & -1 & 0 \\ 
            -1 & 1 & 0 
        \end{bmatrix},\\ 
        A_{2}(x_{1}, x_{3})
        = 
        \begin{bmatrix} 
            0 & 1 & -1 \\ 
            0 & -1 & 1 
        \end{bmatrix}, \\ 
        A_{3}(x_1,x_2)
        = 
        \begin{bmatrix} 
            -x_{2}^{(1)} & x_{1}^{(1)}+x_{2}^{(1)}  & -x_{1}^{(1)} \\ 
            x_{2}^{(1)} & -(x_{1}^{(1)}+x_{2}^{(1)}) & x_{1}^{(1)} 
        \end{bmatrix},
    \end{cases} 
\end{equation*}
and
\begin{equation}
    \label{eq:BI-original}
    x = Q(y), \quad \dot{y} = A(x)u, \quad x \in \Delta_{2}\times \Delta_{2} \times \Delta_{2}, \, u \in \Delta_{3}
\end{equation}
with
\begin{equation*}
    A(x)
    =
    \begin{bmatrix}
        A_{1}(x_{1}, x_{2})\\
        A_{2}(x_{1}, x_{3})\\
        A_{3}(x_{2}, x_{2})
    \end{bmatrix}.
\end{equation*}
We now verify the two conditions of Theorem~\ref{thm:N-player-controllability-1} using the
vector fields
\begin{equation*}
    \begin{aligned}
        \eta_1(x)&=DQ(\nabla h(x))(a_1(x)-a_3(x)),\\
        \eta_2(x)&=DQ(\nabla h(x))(a_2(x)-a_3(x)).
    \end{aligned}
\end{equation*}
The first condition is easily satisfied, since $u_0=(1/3,1/3,1/3)\in\Delta_3^\circ$ 
is uniformly neutralizing.

To verify Condition~2, write
\[
x_i=\bigl(x_i^{(1)},x_i^{(2)}\bigr),\qquad
\xi_i:=\frac12\bigl(x_i^{(1)}-x_i^{(2)}\bigr),\quad i=1,2,3.
\]
Since each learner has two pure strategies, each tangent space
\[
H_i=\{(r,-r):r\in\mathbb R\}
\]
is one-dimensional. Hence, for each block there exists a smooth
positive scalar function $\lambda_i(\xi_i)>0$ such that
\[
P_{H_i}DQ_i(\nabla h_i(x_i))
\begin{bmatrix}
1\\-1
\end{bmatrix}
=
\lambda_i(\xi_i)
\begin{bmatrix}
1\\-1
\end{bmatrix}.
\]

From the matrices in the example, the three columns of $A(x)$ are
\begin{equation*}
    \begin{aligned}
        a_1(x)&=
\begin{bmatrix}
1 &-1&0&0&-x_2^{(1)}& x_2^{(1)}
\end{bmatrix}^\T,
\\
 a_2(x)&=
\begin{bmatrix}
-1&1&1&-1&x_1^{(1)}+x_2^{(1)}&-(x_1^{(1)}+x_2^{(1)})
\end{bmatrix}^\T,
\\
 a_3(x)&=
\begin{bmatrix}
0&0&-1&1&-x_1^{(1)}&x_1^{(1)}
\end{bmatrix}^\T.
    \end{aligned}
\end{equation*}
Therefore, in the local coordinates $\xi=(\xi_1,\xi_2,\xi_3)$,
the push-forward vector fields become
\[
\eta_1(\xi)=
\begin{bmatrix}
\lambda_1(\xi_1)\\[2mm]
\lambda_2(\xi_2)\\[2mm]
(\xi_1-\xi_2)\lambda_3(\xi_3)
\end{bmatrix},
\eta_2(\xi)=
\begin{bmatrix}
-\lambda_1(\xi_1)\\[2mm]
2\lambda_2(\xi_2)\\[2mm]
\left(2\xi_1+\xi_2+\frac32\right)\lambda_3(\xi_3)
\end{bmatrix}.
\]
A direct computation gives
\[
[\eta_1,\eta_2](\xi)=
\begin{bmatrix}
0\\[2mm]
0\\[2mm]
3\bigl(\lambda_1(\xi_1)+\lambda_2(\xi_2)\bigr)\lambda_3(\xi_3)
\end{bmatrix}.
\]
Since $\lambda_i(\xi_i)>0$ for all $i$, the vectors
$\eta_1(\xi)$, $\eta_2(\xi)$, and $[\eta_1,\eta_2](\xi)$ are
linearly independent for every $\xi\in(-1/2,1/2)^3$. Hence $\operatorname{rank}\bigl(\operatorname{Lie}\{\eta_1,\eta_2\}\bigr)=3
=\sum_{i=1}^3 n_i-3$. 
By Theorem~\ref{thm:N-player-controllability-1}, induced primal system and thus \eqref{eq:BI-original} is steerable on $\Delta_2^\circ \times \Delta_2^\circ \times \Delta_2^\circ$.

Interestingly, when all learners adopt FTRL with $h(x) = 1/2 \|x\|^{2}$, we can adopt the following change of variables:
\begin{equation*}
    \begin{cases}
        \xi_{1} = \frac{1}{2}(x_{1}^{(1)} - x_{1}^{(2)}),\\
        \xi_{2} = \frac{1}{2}(x_{2}^{(1)} - x_{2}^{(2)}),\\
        \xi_{3} = \frac{1}{2}(x_{3}^{(1)} - x_{3}^{(2)}),\\
        w_{1} = u_{1} - u_{2},\\
        w_{2} = u_{2} - u_{3}.
    \end{cases}
\end{equation*}
The resulting system is exactly the Brockett's Integrator:
\begin{equation*}
    \dot{\xi}
    =
    \begin{bmatrix}
        1 &0 \\
        0 &1 \\
        -\xi_{2} &\xi_{1}
    \end{bmatrix}
    w, \quad w \in W,
\end{equation*}
where the control set $W$ is the convex polygon
\begin{equation*}
    W=\left\{(v_1,v_2)\in\mathbb{R}^2:\;
        \begin{array}{l}
            1+2v_1+v_2\ge 0,\\[1pt]
            1-v_1+v_2\ge 0,\\[1pt]
            1-v_1-2v_2\ge 0
        \end{array}
    \right\}.
\end{equation*}
\subsubsection{Regulated Matching Pennies}
\label{ex:RMP}
Consider a three player game with two learners and one controller, in which the two learners each have two pure actions and the controller has three. For any fixed controller action, the two learners are involved in a two player zero-sum game, with a payoff matrix dependent on the controller's action. Specifically, the three actions of the controller correspond to the following payoff matrices:
\begin{equation*}
    B_{1} = 
    \begin{bmatrix}
        1 &1 \\
        0 &0
    \end{bmatrix},
    B_{2} = 
    \begin{bmatrix}
        2 &-5 \\
        -3 &2
    \end{bmatrix},
    B_{3} = 
    \begin{bmatrix}
        0 &1 \\
        0 &1
    \end{bmatrix}.
\end{equation*}
Figure~\ref{fig:RMP} illustrates the structure of this game. Denoting $x_{1} = (\alpha, 1-\alpha)$ and $x_{2} = (\beta, 1-\beta)$, we can model this game in the form of \eqref{eq:n-player-raw-block} with
\begin{equation*}
    \begin{aligned}
        A_{1}(x_{2}) 
    &=
    \begin{bmatrix}
        B_{1}x_{2}, B_{2}x_{2}, B_{3}x_{2}
    \end{bmatrix},\\
    A_{2}(x_{1}) 
    &=
    \begin{bmatrix}
        -B_{1}x_{1}, -B_{2}x_{1}, -B_{3}x_{1}
    \end{bmatrix}
    \end{aligned}
\end{equation*}
Suppose that the learning agents adopt the replicator dynamics. 
For $\bar u=(1/3,1/3,1/3)$, the drift field is
\[
\eta_0(x)=\frac13 DQ(\nabla h(x))A(x)\mathbf 1_3.
\]
Under replicator dynamics and the reduced coordinates $(\alpha,\beta)$, this becomes
\[
\eta_0(\alpha,\beta)=
\begin{bmatrix}
2\alpha(1-\alpha)(2\beta-1)\\
2\beta(1-\beta)(1-2\alpha)
\end{bmatrix},
\]
which is the classical matching-pennies replicator field and is periodic on $(0,1)^2$. This satisfies condition 1) of Theorem~\ref{thm:N-player-controllability-2}.

\begin{figure}[t]
    \centering
    \resizebox{\columnwidth}{!}{%
        \begin{tikzpicture}[
            >=Stealth,
            font=\footnotesize,
            node distance=1.15cm and 2.6cm,
            player/.style={
                draw, rounded corners, align=center,
                inner sep=3.5pt, minimum width=2.75cm, minimum height=0.95cm
            },
            block/.style={
                draw, rounded corners, align=center,
                inner sep=3.0pt, minimum width=3.2cm, minimum height=0.85cm
            },
            lab/.style={font=\scriptsize, fill=white, inner sep=1.0pt},
            ctrl/.style={->, dashed},
            dep/.style={->}
            ]

            \node[player] (C) {\textbf{Controller}\\ $u\in\Delta^3$};
            \node[block, below=0.95cm of C] (M) {Matrix selection\\ {\scriptsize $B^{k},\ k\in\{1,2,3\}$}};

            \node[player, below left=1.10cm and 2.0cm of M] (L1) {\textbf{Learner 1}\\ $x_1\in\Delta^2$ \\$p_{1} = B_{k}x_{2}$};
            \node[player, below right=1.10cm and 2.0cm of M] (L2) {\textbf{Learner 2}\\ $x_2\in\Delta^2$ \\$p_{2} = -B_{k}^{\T}x_{1}$};

            \draw[ctrl] (C) -- node[lab, right] {$u$} (M);

            \draw[ctrl] (M) -- node[lab, sloped, above] {$B^{k}$} (L1);
            \draw[ctrl] (M) -- node[lab, sloped, above] {$B^{k}$} (L2);

            \draw[dep] (L2) to[bend left=10] node[lab, above] {$x_2$} (L1);
            \draw[dep] (L1) to[bend left=10] node[lab, below] {$x_1$} (L2);

        \end{tikzpicture}
    }
    \caption{Interdependency graph for the Regulated Matching Pennies game: This figure shows the interdependency among players in the Regulated Matching Pennies game presented in Section~\ref{ex:RMP}. The two learners are involved in a two-player zero-sum game with a payoff matrix chosen by the controller. Again, we use $u \in \Delta_{3}$ to represents the controller's mixed strategy, $x_{i} \in \Delta_{2}$  represents the mixed strategy of learner $i$, and $p_{i}$ represents the payoff vector of learner $i$.}
    \label{fig:RMP}
\end{figure}
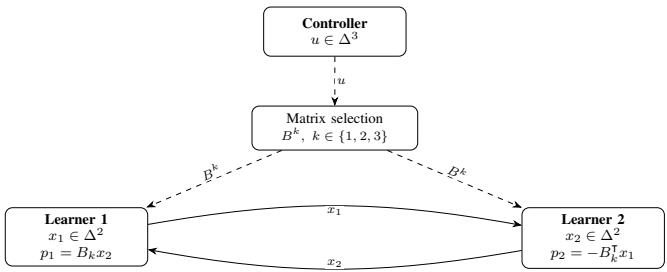

We now verify Condition~2). Taking again
\begin{equation*}
    \begin{aligned}
        \eta_1(x)&=DQ(\nabla h(x))(a_1(x)-a_3(x)),
\\
\eta_2(x)&=DQ(\nabla h(x))(a_2(x)-a_3(x)),
    \end{aligned}
\end{equation*}
where $a_1,a_2,a_3$ are the columns of $A(x)$. A direct
computation gives
\[
a_1(x)=
\begin{bmatrix}
1\\
0\\
-\alpha\\
-\alpha
\end{bmatrix},
a_2(x)=
\begin{bmatrix}
7\beta-5\\
2-5\beta\\
3-5\alpha\\
7\alpha-2
\end{bmatrix},
a_3(x)=
\begin{bmatrix}
1-\beta\\
1-\beta\\
0\\
-1
\end{bmatrix}.
\]
For replicator dynamics, in the reduced coordinates
$(\alpha,\beta)$ these become
\[
\eta_1(\alpha,\beta)=
\begin{bmatrix}
\alpha(1-\alpha)\\[1mm]
-\beta(1-\beta)
\end{bmatrix},
\eta_2(\alpha,\beta)=
\begin{bmatrix}
(12\beta-7)\alpha(1-\alpha)\\[1mm]
(4-12\alpha)\beta(1-\beta)
\end{bmatrix}.
\]
Their Lie bracket is
\[
[\eta_1,\eta_2](\alpha,\beta)
=
-12\alpha\beta(1-\alpha)(1-\beta)
\begin{bmatrix}
1\\
1
\end{bmatrix}.
\]
Therefore,
\begin{equation*}
    \begin{aligned}
        &\det\bigl[\eta_1(\alpha,\beta),[\eta_1,\eta_2](\alpha,\beta)\bigr] \\
=&
-12\alpha\beta(1-\alpha)(1-\beta)
\Bigl(\alpha(1-\alpha)+\beta(1-\beta)\Bigr).
    \end{aligned}
\end{equation*}
Since $(\alpha,\beta)\in(0,1)^2$, the right-hand side is never
zero. Thus
\[
\operatorname{rank}\bigl(\operatorname{Lie}\{\eta_1,\eta_2\}\bigr)=2
=
(2-1)+(2-1)
\]
for every $(\alpha,\beta)\in \Delta_2^\circ\times\Delta_2^\circ$.
Hence Theorem~\ref{thm:N-player-controllability-2} applies, and the induced primal system
\eqref{eq:n-player-primal} is controllable on
$\Delta_2^\circ \times \Delta_2^\circ$.

\section{Conclusion and Further Discussions}
In this paper, we show that steering continuous-time FTRL learners in finite games can be understood as a control problem on the relative interior of a simplex, or a product of simplices. From this perspective, the issue is not only how learners adapt, but when a model-aware player can guide that adaptation to a desired interior strategy configuration through standard strategic play alone, without changing the game’s payoff structure. In the two-player case, we characterize this exactly through the existence of a fully mixed neutralizing strategy and a rank condition on the projected payoff map. In settings with multiple learners, we identify two sufficient routes to controllability: one based on uniform neutralization, and one based on periodic drift together with a Lie-algebra rank condition.

These results followed from the insight that steering FTRL dynamics has the structure of a geometric control problem, allowing us to bring ideas of controllability into the analysis of vulnerabilities in strategic learning and identify how the geometry of the game determines the influence an agent can have on the overall learning behavior.

\printbibliography

\begin{IEEEbiography}[{\includegraphics[width=1in, height=1.25in, clip, keepaspectratio]{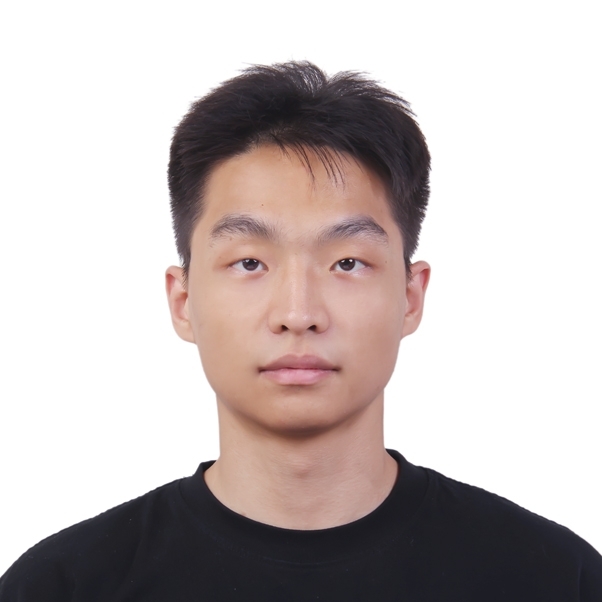}}]{Heling Zhang} (Graduate Student Member, IEEE) received his B.S. and M.S. degrees in Electrical and Computer Engineering from the University of Illinois at Urbana-Champaign, Urbana, IL, USA. He is currently pursuing a Ph.D. degree in Electrical and Computer Engineering at the same institution. His research interests include control theory and machine learning.
\end{IEEEbiography}

\begin{IEEEbiography}[{\includegraphics[width=1in, height=1.25in, clip, keepaspectratio]{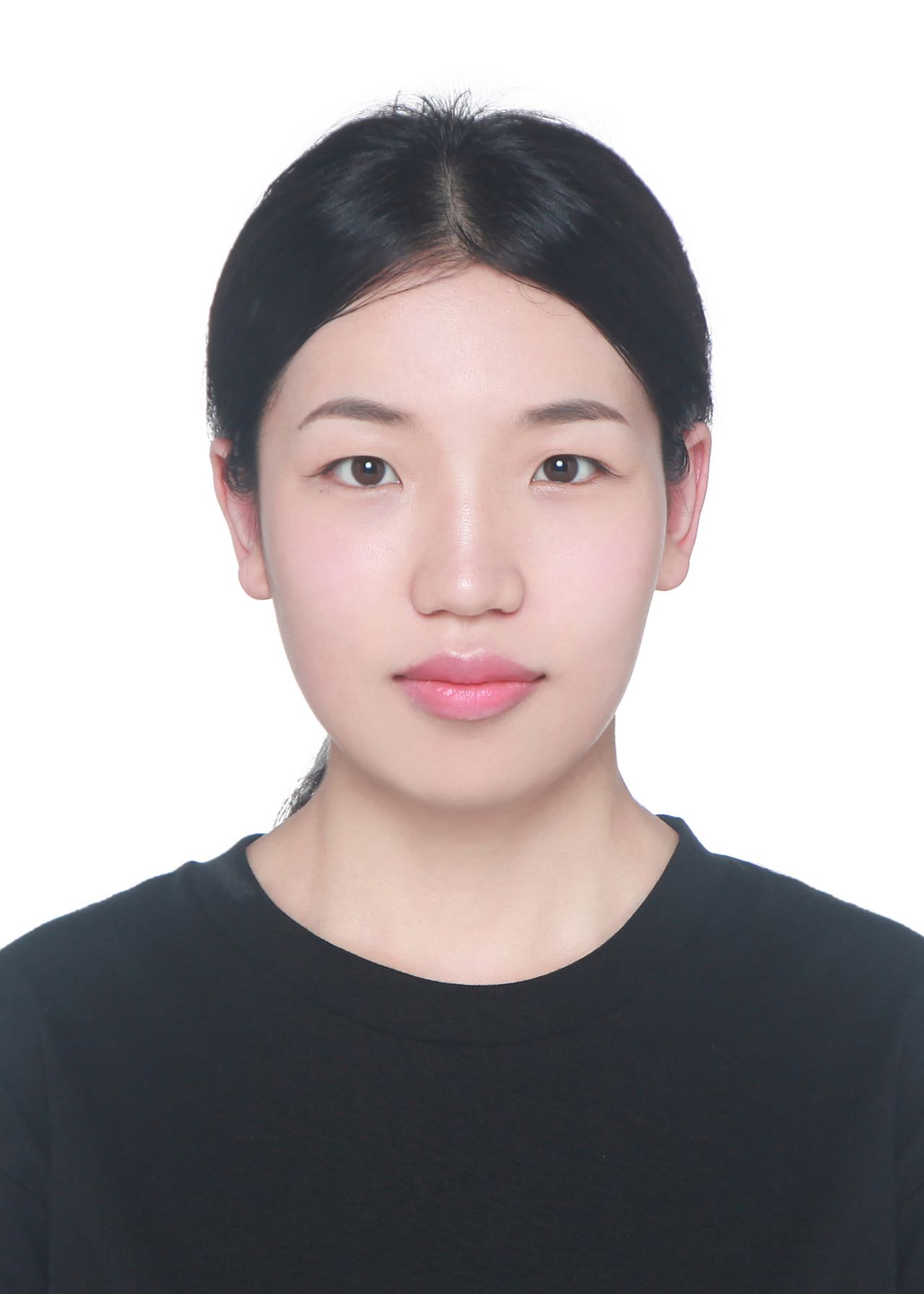}}]{Siqi Du} (Graduate Student Member, IEEE) received the B.S. degree in Management Science from Sichuan University and the M.S. degree in Industrial Engineering from the University of Illinois at Urbana-Champaign. She is currently pursuing the Ph.D. degree in the Department of Industrial and Enterprise Systems Engineering at UIUC. Her research interests include operations research and control systems.
\end{IEEEbiography}

\begin{IEEEbiography}[{\includegraphics[width=1in, height=1.25in, clip, keepaspectratio]{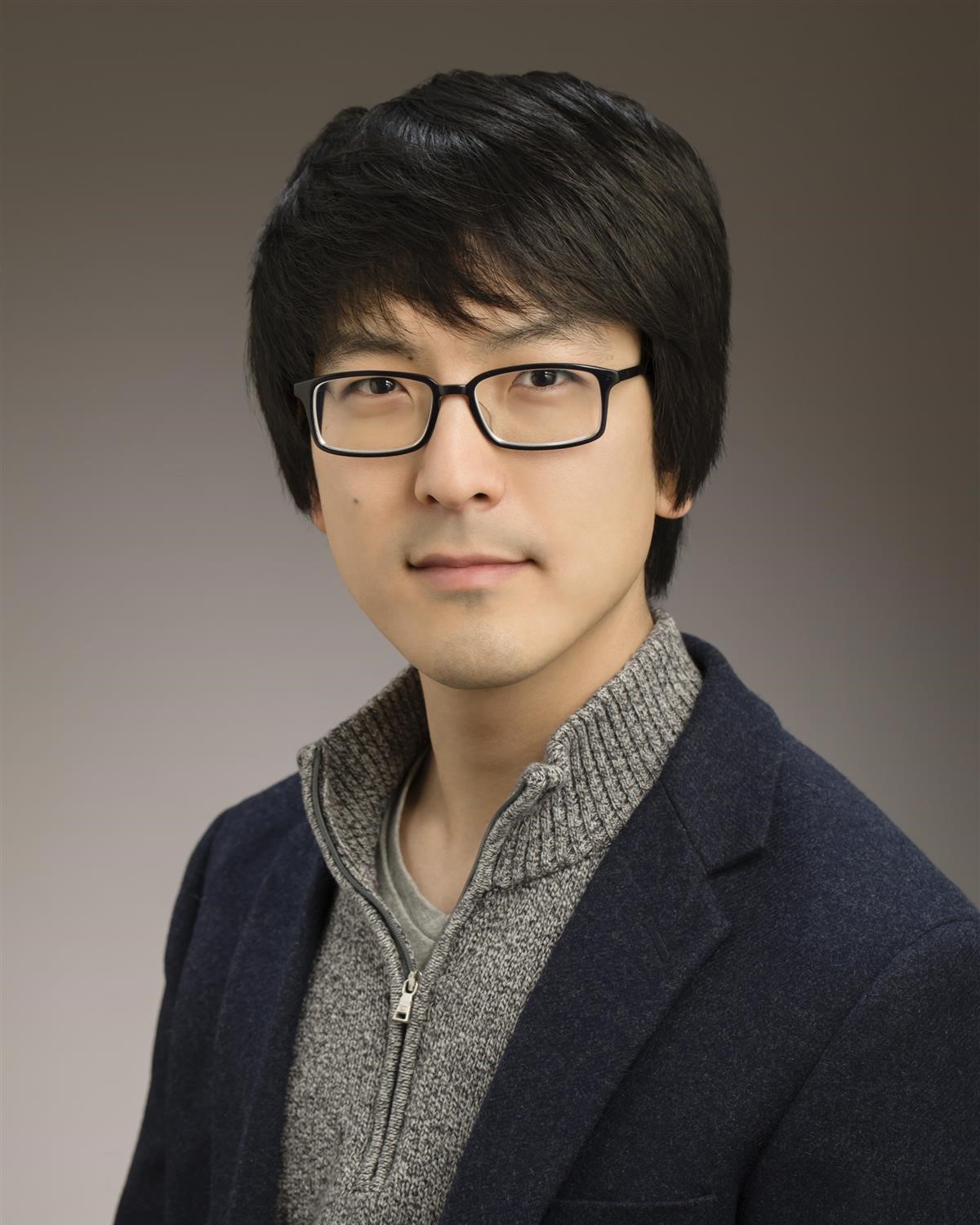}}]{Roy Dong}
(Member, IEEE)
is an Assistant Professor in the Industrial \& Enterprise Systems Engineering department at the University of Illinois at Urbana-Champaign. He received a BS Honors in Computer Engineering and a BS Honors in Economics from Michigan State University in 2010. He received a PhD in Electrical Engineering and Computer Sciences at the University of California, Berkeley in 2017, where he was funded in part by the NSF Graduate Research Fellowship. Prior to his current position, he was a postdoctoral researcher in the Berkeley Energy \& Climate Institute, a visiting lecturer in the Industrial Engineering and Operations Research department at UC Berkeley, and a Research Assistant Professor in the Electrical and Computer Engineering department at the University of Illinois at Urbana-Champaign. His research uses tools from control theory, economics, statistics, and optimization to understand the closed-loop effects of machine learning, with applications in cyber-physical systems such as the smart grid, modern transportation networks, and autonomous vehicles.
\end{IEEEbiography}

\end{document}